\documentclass[12pt]{article}
\usepackage[cp1250]{inputenc}
\usepackage[verbose,a4paper,tmargin=0.5in,lmargin=1in,rmargin=1in]{geometry}
\usepackage[pdftex]{graphicx}
\usepackage{amsfonts}
\usepackage{indentfirst}
\usepackage{latexsym}
\usepackage{amsmath}
\usepackage{amsthm}
\usepackage{amssymb}
\usepackage{psfrag}
\usepackage[mathscr]{eucal} 
\usepackage{color}
\usepackage{cite}
\usepackage{longtable}
\usepackage{verbatim}

\usepackage{graphicx}
\usepackage[all]{xy}
\usepackage{float}

\newtheorem{Twierdzenie}{Theorem}[section]
\newtheorem{Wniosek}{Corollary}[section]
\newtheorem{Definicja}{Definition}[section]
\newtheorem{Lemat}{Lemma}[section]

\newfloat{Scheme}{tbp}{muj}

\title{Two-sided conformally recurrent self-dual spaces.}

\author{$\textrm{Adam Chudecki}^{*}$}

\begin{document}

\maketitle

$*$ Center of Mathematics and Physics, Lodz University of Technology, 
\newline
$\ \ \ \ \ $ Al. Politechniki 11, 90-924 Łódź, Poland, adam.chudecki@p.lodz.pl
\newline
\newline
\newline
\textbf{Abstract}. 
\newline
Two-sided conformally recurrent 4-dimensional self-dual spaces are considered. It is shown that such spaces are equipped with nonexpanding congruences of null strings. The general structure of weak nonexpanding hyperheavenly spaces is given. Finally, the general metrics of Petrov-Penrose type $[\textrm{D}] \otimes [-]$ spaces are presented.
%\newline
%\newline
%\textbf{PACS numbers:} 04.20.Cv, 04.20.Jb, 04.20.Gz

%#####################################################################################

%\renewcommand{\arraystretch}{1.5}
%\setlength\arraycolsep{2pt}
\setcounter{equation}{0}

\section{Introduction}

The paper is devoted to some interesting aspects of a complex geometry. By "complex geometry" we understand the geometry of the manifolds which are 4-dimensional in a complex sense. We assume, that such manifolds are equipped with a holomorphic metric. Hence, they are generalizations of 4-dimensional real manifolds equipped with a real smooth metric. It is well-known that in dimension 4 there are three different types of real manifolds. These are Lorentzian manifolds (in this case the metric has signature $(+++-)$), neutral manifolds (signature of the metric is $(++--)$) and Riemannian (also called proper-Riemannian or Euclidean) manifolds (signature $(++++)$). 

The spaces which we analyze in this paper have additional property: they are \textsl{two-sided conformally recurrent}. The idea of the spaces which are \textsl{conformally recurrent}, i.e., for which $\nabla_{m} C_{abcd} = r_{m} C_{abcd}$\footnote{For \textsl{conformally symmetric spaces} $r_{m}=0$.}, $C_{abcd} \ne 0$, has a long history in general theory of relativity and in differential geometry. One of the most transparent papers devoted to such spaces is that by 
McLenaghan and Leroy \cite{Mclenaghan_Leroy}. In this paper the authors have found all Lorentzian metrics which are two-sided conformally recurrent. It appeared that if such a space is not conformally flat its metric is of the Petrov-Penrose type $[\textrm{D}]$ or $[\textrm{N}]$.

What more can be done in the subject if the metrics of the conformally recurrent spaces are explicitly known? The answer is: the results can be generalized to the complex case and to the real cases of the different signatures. Although real Riemannian spaces and real neutral spaces are not the realistic models of space-time, they play a great role in theoretical physics and in geometry. Neutral spaces appear in Walker and Osserman geometries \cite{Walker, Chudecki_Przanowski_Walkery}, projective structures \cite{Dunajski_Mettler}, integrable systems and ASD structures \cite{Dunajski_Tod}, rolling bodies, para-Hermite and para-K\"{a}hler spaces \cite{Nurowski_1, Nurowski_2}. The investigations of real Riemannian spaces led to the idea of \textsl{gravitational instantons} \cite{Przanowski_inst_1, Przanowski_inst_2,Dunajski_Tod_2,Tod}. 4-dimensional real manifolds can be obtained from a complex solution by the procedure of \textsl{real slice} of complex metric \cite{Rozga}. This remarkable fact justifies the studies of complex spaces in dimension 4.

Generalization of the Lorentzian conformally recurrent spaces has been presented in the distinguished paper by Plebański and Przanowski \cite{Plebanski_Przanowski_rec}. They extended the definition of the conformal recurrent spaces to the \textsl{two-sided conformally recurrent spaces} and they proved that if a space is two-sided conformally recurrent then both self-dual (SD) and anti-self-dual (ASD) Weyl spinors must be of the Petrov-Penrose types $[\textrm{D}]$, $[\textrm{N}]$ or $[-]$. Hence, the only complex spaces which can be two-sided conformally recurrent are spaces of the types $[\textrm{D}] \otimes [\textrm{D}]$, $[\textrm{N}] \otimes [\textrm{N}]$, $[\textrm{D}] \otimes [-]$ and $[\textrm{N}] \otimes [-]$. The generalization of the results of \cite{Mclenaghan_Leroy} led to the complex spaces of the types $[\textrm{D}] \otimes [\textrm{D}]$ and $[\textrm{N}] \otimes [\textrm{N}]$. Plebański and Przanowski found neutral and Riemannian slices of these types of spaces. 

An important section of \cite{Plebanski_Przanowski_rec} is devoted to the self-dual spaces. Self-dual spaces are the spaces for which ASD Weyl spinor vanishes. SD metrics of the type $[\textrm{N}] \otimes [-]$ have been explicitly presented in \cite{Plebanski_Przanowski_rec}, but Plebański and Przanowski wrote: "Up to now we have not succeeded in integrating the type $[\textrm{D}] \otimes [-]$. It seems to be a rather hard problem". The main aim of our paper is to fill this gap.

In our paper the formalism of \textsl{weak hyperheavenly ($\mathcal{HH}$) spaces} is used \cite{Chudecki_Przanowski_Walkery, Plebanski_Rozga}. The important property of the weak $\mathcal{HH}$-spaces is the existence of 2-dimensional, totally null completely integrable distribution. The family of the integral manifolds of such a distribution is called \textsl{congruence (foliation) of null strings} (for definition, see section \ref{section_congruences_of_null_strings}). Two-sided conformally recurrent type $[\textrm{D}] \otimes [-]$  spaces are equipped with two SD congruences and both of them are parallely propagated. This fact is crucial for our considerations. 

The main results of our paper are the metrics (\ref{postac_dKS}), (\ref{Non_Einstein_second_solution}) and (\ref{rozwiazanie_Dxnic_einsteinowskie_bez_funkcji_f}). These are general metrics of the SD spaces of the type $[\textrm{D}]^{nn} \otimes [-]^{e}$. They admit real slices. Real neutral slices equipped with two nonexpanding congruences of (real) SD null strings can be obtained immediately by considering all coordinates and constants as real ones. Thus, in neutral case the metrics (\ref{postac_dKS}), (\ref{Non_Einstein_second_solution}) and (\ref{rozwiazanie_Dxnic_einsteinowskie_bez_funkcji_f}) are general metrics of so called para-K\"{a}hler type $[\textrm{D}]^{nn} \otimes [-]^{e}$ spaces. Real Riemannian slices involve a more subtle approach. We analyze Einstein metric (\ref{rozwiazanie_Dxnic_einsteinowskie_bez_funkcji_f}) and we transform it to the form (\ref{typ_DnnxDnn_einstein_IIformalizm_optymalna}) which is more plausible for obtaining real Riemannian slices. It appears that Riemannian slice of the metric (\ref{typ_DnnxDnn_einstein_IIformalizm_optymalna}) is Fubini-Study metric. Riemannian slices of the non-Einstein metrics (\ref{postac_dKS}) and (\ref{Non_Einstein_second_solution}) will be presented elsewhere.

The paper is organized, as follows.

In section 2 short introduction to the formalism used in the paper is given. The concept of the congruences of null strings and Petrov-Penrose classification of the Weyl spinors is introduced. It is also proven that two-sided conformally recurrent spaces are equipped with nonexpanding congruence of null strings. Section 3 is devoted to the nonexpanding weak hyperheavenly spaces. The definition and the form of the metric of such spaces are introduced. It is explained how to pass from such spaces to the algebraically degenerated SD spaces.

In section 4 a special subtype of the SD spaces of the type [D], namely, type $[\textrm{D}^{nn}] \otimes [-]^{e}$ is considered. Spaces of such a type are two-sided conformally recurrent. The final equations (\ref{koncowy_uklad}) are found and integrated. The solutions led to the metrics (\ref{postac_dKS}), (\ref{Non_Einstein_second_solution}) and (\ref{rozwiazanie_Dxnic_einsteinowskie_bez_funkcji_f}). Finally, the transformation of the Einstein solution (\ref{rozwiazanie_Dxnic_einsteinowskie_bez_funkcji_f}) to the double null coordinates system is presented. 

Concluding remarks close the paper.

\setcounter{equation}{0}
\section{Preliminaries}

\subsection{Petrov-Penrose classification}

Let $(\mathcal{M}, ds^{2})$ be a 4-dimensional complex analytic differential manifold equipped with a holomorphic metric. The metric $ds^{2}$ can be written in the form 
\begin{equation}
ds^{2} = 2e^{1}e^{2} + 2e^{3}e^{4} 
\end{equation}
where 1-forms $(e^{1},e^{2},e^{3},e^{4})$ are the members of a \textsl{complex null tetrad} and they form a basis of 1-forms. The dual basis is denoted by $(\partial_{1}, \partial_{2}, \partial_{3}, \partial_{4})$ and it is also called \textsl{the null tetrad}. $g^{A\dot{B}}$ and $\partial_{A\dot{B}}$ are the spinorial images of the null tetrads
\begin{equation}
\label{definicccja_gAB}
(g^{A\dot{B}}) := \sqrt{2}
\left[\begin{array}{cc}
e^4 & e^2 \\
e^1 & -e^3
\end{array}\right] 
 , \ 
(\partial_{A\dot{B}}) := -\sqrt{2}
\left[\begin{array}{cc}
\partial_{4} & \partial_{2} \\
\partial_{1} & -\partial_{3}
\end{array}\right] , \  A=1,2,  \ \dot{B}=\dot{1},\dot{2}
\end{equation}
Spinorial indices are manipulated according to the following rules 
\begin{equation}
\label{spinorial_indices_lowering_rule}
m_{A} = \ \in_{A B} m^{B}
\ , \ \ \ 
m^{A} = m_{B} \in^{BA}
\ , \ \ \
m_{\dot{A}} = \ \in_{\dot{A} \dot{B}} m^{\dot{B}}
\ , \ \ \ 
m^{\dot{A}} = m_{\dot{B}} \in^{\dot{B} \dot{A}}
\end{equation}
where $\in_{AB}$ and $\in_{\dot{A}\dot{B}}$ are the spinor Levi-Civita symbols
\begin{eqnarray}
&& (\in_{AB})  := \left[ \begin{array}{cc}
                            0 & 1   \\
                           -1 & 0  
                            \end{array} \right] =:  (\in^{AB} )
\ \ \ , \ \ \ 
 (\in_{\dot{A}\dot{B}})  := \left[ \begin{array}{cc}
                            0 & 1   \\
                           -1 & 0  
                            \end{array} \right] =:  (\in^{\dot{A}\dot{B}} ) 
\\ \nonumber
&& \in_{AC} \in^{AB} = \delta^{B}_{C} \ \ \ , \ \ \ \in_{\dot{A}\dot{C}} \in^{\dot{A}\dot{B}} = \delta^{\dot{B}}_{\dot{C}} \ \ \ , \ \ \ 
(\delta^{A}_{C})= (\delta^{\dot{B}}_{\dot{C}})= \left[ \begin{array}{cc}
                            1 & 0   \\
                            0 & 1  
                            \end{array} \right]
\end{eqnarray}
Spinorial formalism used in this paper is Infeld - Van der Waerden - Plebański notation. For more details see \cite{Plebanski_Spinors,Pleban_formalism_2}. Also, a brief summary of this formalism has been presented in \cite{Chudecki_struny,Plebanski_Przanowski_rec}. 

Spinorial image of the SD part of the Weyl tensor is called \textsl{the SD Weyl spinor} and it is 4-index, undotted spinor totally symmetric in all indices $C_{ABCD} = C_{(ABCD)}$. It is well know that it can be presented as a product of some 1-index spinors
\begin{equation}
\label{decomposition_of_SD_Weyl}
C_{ABCD} = a_{(A} b_{B} c_{C} d_{D)}
\end{equation}
$a_{A}$, $b_{A}$, $c_{A}$ and $d_{A}$ are undotted complex spinors which are called \textsl{Penrose spinors}. If all Penrose spinors are mutually linearly independent then the SD conformal curvature is \textsl{algebraically general}. If at least two Penrose spinors are proportional to each other, then the SD conformal curvature is \textsl{algebraically special}. All possible degenerations between Penrose spinors give the Petrov-Penrose classification of the SD conformal curvature:
\begin{eqnarray}
\textrm{type [I]} : \ && \ C_{ABCD} = a_{(A} b_{B} c_{C} d_{D)}  
\\ \nonumber
\textrm{type [II]}: \ && \ C_{ABCD} = a_{(A} a_{B} b_{C} c_{D)}   
\\ \nonumber
\textrm{type [D]}: \ && \  C_{ABCD} = a_{(A} a_{B} b_{C} b_{D)}  
\\ \nonumber
\textrm{type [III]}: \ && \  C_{ABCD} = a_{(A} a_{B} a_{C} b_{D)}  
\\ \nonumber
\textrm{type [N]}: \ && \  C_{ABCD} = a_{A} a_{B} a_{C} a_{D} 
\\ \nonumber
\textrm{type }[-]: \ && \  C_{ABCD} = 0 
\end{eqnarray} 
The similar classification holds true for the ASD Weyl spinor $C_{\dot{A}\dot{B}\dot{C}\dot{D}}$. 

In a complex case SD and ASD Weyl spinors are unrelated and they can be of an arbitrary Petrov-Penrose type. To find the type of the conformal curvature of a complex space one has to determine the Petrov-Penrose type of SD and ASD parts separately. Commonly used symbol is $[\textrm{SD type}] \otimes [\textrm{ASD type}]$. For example,  $[\textrm{D}] \otimes [\textrm{N}]$ means that SD Weyl spinor is of the type [D] while ASD Weyl spinor is of the type [N].

\subsection{Congruences of null strings}
\label{section_congruences_of_null_strings}

In this section we equip $(\mathcal{M}, ds^{2})$ with additional structure: 2-dimensional, integrable, totally null distribution. Family of integral manifolds of such distribution is called \textsl{congruence of null strings} and it plays a fundamental role in our further investigations. For more detailed analysis of the congruences of null strings see \cite{Plebanski_Rozga,Chudecki_struny}.

Let $\mathcal{D}_{m^{A}} = \{ m_{A} a_{\dot{A}}, m_{A} b_{\dot{A}} \}$, $a_{\dot{A}} b^{\dot{B}} \ne 0$ be a 2-dimensional SD holomorphic distribution defined by the Pfaff system 
\begin{equation}
m_{A} g^{A \dot{B}} = 0
\end{equation}
The distribution $\mathcal{D}_{m^{A}}$ is integrable in the Frobenius sense if and only if the spinor field $m_{A}$ satisfies the equations
\begin{equation}
\label{rownania_strun_SD}
m^{A} m^{B} \nabla_{A \dot{C}} m_{B} = 0
\end{equation}
Eqs. (\ref{rownania_strun_SD}) are called \textsl{SD null string equations}. If Eqs. (\ref{rownania_strun_SD}) hold true one says that the spinor $m_{A}$ generates the \textsl{congruence of SD null strings}. The integral manifolds of the distribution $\mathcal{D}_{m^{A}}$ are totally null and geodesic, 2-dimensional SD holomorphic surfaces (\textsl{SD null strings}). The family of such surfaces constitute \textsl{the congruence of SD null strings}. More precisely 
\begin{Definicja}
\label{definicja_foliacji_strun}
A congruence (foliation) of null strings in a complex 4-dimensional manifold $\mathcal{M}$ is a family of totally null and totally geodesic 2-dimensional holomorphic surfaces such that for every point $p\in \mathcal{M}$ there exists only one surface of this family such that $p$ belongs to this surface.
\end{Definicja}
Eqs. (\ref{rownania_strun_SD}) are equivalent to the equations
\begin{equation}
\label{rozwiniete_rownania_strun}
\nabla_{A \dot{C}} m_{B} = Z_{A \dot{C}} m_{B} + \in_{AB} M_{\dot{C}}
\end{equation}
where $Z_{A \dot{C}}$ is \textsl{the Sommers vector} and $M_{\dot{C}}$ is \textsl{the expansion of the congruence}\footnote{The concept of the expansion of the congruence of null strings should not be mistaken with the concept of the expansion of the congruence of null geodesic lines. The expansion of the congruence of null geodesics is well-known in general theory of relativity and together with other \textsl{optical scalars} like \textsl{twist} and \textsl{shear} is used for the classification of the exact solutions of Einstein field equations, see, e.g., \cite{Exacty}.}. The expansion is the most important property of the congruence of null strings. If $M_{\dot{A}} \ne 0$ the congruence is \textsl{expanding}; if $M_{\dot{A}} = 0$ then the congruence is \textsl{nonexpanding}. If the congruence is nonexpanding it means that the distribution $\mathcal{D}_{m^{A}}$ is parallely propagated, i.e., $\nabla_{X}V \in \mathcal{D}_{m^{A}}$ for every vector field $V \in \mathcal{D}_{m^{A}}$ and for arbitrary vector field $X$. 

Note, that if $M_{\dot{C}}=0$ then from (\ref{rozwiniete_rownania_strun}) it follows that covariant derivative of the spinor field $m_{B}$ is proportional to this field, $\nabla_{a} m_{B} = Z_{a} m_{B}$. Such a spinor field is called \textsl{recurrent}. Hence
\begin{Wniosek}
Nonexpanding congruences of null strings are generated by the recurrent spinor fields. \hfill $\blacksquare$ 
\end{Wniosek}

Two facts are important for further analysis (see \cite{Chudecki_struny} for proofs)
\begin{Twierdzenie} 
\label{theoorem_1}
If a spinor $m_{A}$ generates a congruence of SD null strings, then it is a Penrose spinor. \hfill $\blacksquare$ 
\end{Twierdzenie}
\begin{Twierdzenie} 
If a spinor $m_{A}$ generates a nonexpanding congruence of SD null strings, then it is a multiple Penrose spinor. \hfill $\blacksquare$ 
\end{Twierdzenie}
From Theorem \ref{theoorem_1} it follows, that maximal number of distinct congruences of SD (ASD) null strings is  4. Only SD (ASD) type $[\textrm{I}]$ can be equipped with 4 distinct congruences of SD (ASD) null strings. For our purposes it is especially desirable to analyze Petrov-Penrose type $[\textrm{D}]$. There are 6 possibilities (see Scheme \ref{Degeneration_scheme_of_complex_case})\footnote{In Einstein spaces only subtypes $[\textrm{D}]^{nn}$ and $[\textrm{D}]^{ee}$ are admitted \cite{Plebanski_surf,Przanowski_Formanski_Chudecki}.}. Superscript \textsl{e} means that the corresponding congruence is expanding while superscript \textsl{n} means that it is nonexpanding. The number of superscripts corresponds to the number of distinct congruences.

[\textbf{Remark}. In real spaces equipped with a neutral signature metric there are two different types [D] of SD Weyl spinor. The first one is usually denoted by $[\textrm{D}_{r}]$ (subscript $r$ means \textsl{real}). In this case SD Weyl spinor can be decomposed according to the formula $C_{ABCD} = a_{(A} a_{B} b_{C} b_{D)}$ where $a_{A}$ and $b_{A}$ are real 1-index undotted spinors. The second type is denoted by $[\textrm{D}_{c}]$ (subscript $c$ means \textsl{complex}) and its decomposition reads $C_{ABCD} = a_{(A} a_{B} \bar{a}_{C} \bar{a}_{D)}$ where $a_{A}$ is a complex 1-index undotted spinor and bar denotes complex conjugation. Hence, Scheme \ref{Degeneration_scheme_of_complex_case} holds true also for the type $[\textrm{D}_{r}]$. Type $[\textrm{D}_{c}]$ does not admit any real congruences of null strings.]

\begin{Scheme}[!ht]
\begin{displaymath}
%\resizebox{1\textwidth}{!}{
\xymatrixcolsep{0.3cm}
\xymatrixrowsep{0.8cm}
\xymatrix{
\textrm{0 congruences:} &  [\textrm{D}]  \ar[d]  \ar[dr]  &  &   \\ 
\textrm{1 congruence:} &  [\textrm{D}]^{e} \ar[d]  \ar[dr] & [\textrm{D}]^{n} \ar[d]  \ar[dr]   \\
\textrm{2 congruences:} &  [\textrm{D}]^{ee}  & [\textrm{D}]^{en}  & [\textrm{D}]^{nn}  \\
}
%}
\end{displaymath} 
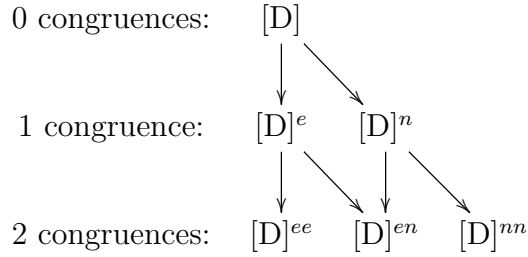
\caption{Subtypes of Petrov-Penrose type $[\textrm{D}]$.}
\label{Degeneration_scheme_of_complex_case}
\end{Scheme}

\begin{Twierdzenie}
\label{Uwaga_1}
If the SD (ASD) Weyl spinor vanishes, then the space is equipped with infinitely many congruences of SD (ASD) null strings. If the curvature scalar is nonzero $R \ne 0$ then all congruences of SD (ASD) null strings are expanding. If the curvature scalar vanishes $R = 0$ then there are both expanding and nonexpanding congruences of SD (ASD) null strings. 
\end{Twierdzenie}
\begin{proof}
See \cite{Plebanski_Rozga}.
\end{proof}
We use the symbol $[\textrm{any}] \otimes [-]^{e}$ for SD spaces with $R \ne 0$ and the symbol $[\textrm{any}] \otimes [-]^{n}$ for SD spaces with $R = 0$.

\subsection{Two-sided conformally recurrent spaces}

The following definitions can be found in \cite{Ruse, Exacty}

\begin{Definicja}
\textsl{Recurrent space} is a nonflat space in which the Riemann tensor satisfies
\begin{equation}
\label{warunek_rekurencyjnosci}
\nabla_{m} R_{abcd} = r_{m} R_{abcd}
\end{equation}
If $r_{m}=0$ then the space is \textsl{symmetric}.
\end{Definicja}
\begin{Definicja}
\textsl{Conformally recurrent space} is space which is not conformally flat and for which the Weyl tensor satisfies
\begin{equation}
\nabla_{m} C_{abcd} = r_{m} C_{abcd}
\end{equation}
If $r_{m}=0$ then the space is \textsl{conformally symmetric}.
\end{Definicja}
It can be shown that (\ref{warunek_rekurencyjnosci}) is equivalent to the following relations
\begin{equation}
  \nabla_{m} C_{abcd} = r_{m} C_{abcd}, \ \nabla_{m} C_{ab} = r_{m} C_{ab}, \ \nabla_{m} R = r_{m} R 
\end{equation}
where $C_{ab}$ is the traceless Ricci tensor and $R$ is the curvature scalar. It implies that each recurrent (symmetric) space is conformally recurrent (conformally symmetric). Finally for the Lorentzian manifolds we have \cite{Exacty}
\begin{Definicja}
\label{Definicja_complex_recurrent}
\textsl{Complex recurrent space-time} is a nonflat space in which the SD Weyl tensor $C_{abcd}^{*}$ satisfies
\begin{equation}
\nabla_{m} C_{abcd}^{*} = r_{m} C_{abcd}^{*}
\end{equation}
\end{Definicja}
Because in the complex spaces and in the real neutral spaces SD and ASD Weyl spinors are unrelated, the following definition has been formulated in \cite{Plebanski_Przanowski_rec} 
\begin{Definicja} 
\label{Definicja_two_sided_recurrent}
Let $\mathcal{M}$ be a four-dimensional real smooth or complex analytic differential manifold equipped with the real or holomorphic metric $ds^{2}$. Then the pair $(\mathcal{M}, ds^2)$ is two-sided conformally recurrent Riemannian manifold if there exist vectors $r_{m}$ and $\dot{r}_{m}$ such that
\begin{subequations}
\begin{eqnarray}
\label{warunek_rekurecyjnosci_1}
&& \nabla_{m} C_{ABCD} = r_{m} C_{ABCD} 
\\
&& \nabla_{m} C_{\dot{A}\dot{B}\dot{C}\dot{D}} = \dot{r}_{m} C_{\dot{A}\dot{B}\dot{C}\dot{D}}
\end{eqnarray}
\end{subequations}
and $C_{ABCD}$ and $C_{\dot{A}\dot{B}\dot{C}\dot{D}}$ do not vanish simultaneously.
\end{Definicja}
In the case of real Lorentzian spaces, the Definition \ref{Definicja_two_sided_recurrent} and the Definition \ref{Definicja_complex_recurrent} are equivalent.

The following Lemma can be easily extracted from \cite{Plebanski_Przanowski_rec} although it has not been formulated there explicitly.

\begin{Lemat}
\label{Lemat_o_strunach}
Let $C_{ABCD} \ne 0$ and $\nabla_{m} C_{ABCD} = r_{m} C_{ABCD}$. Then $(\mathcal{M}, ds^{2})$ is equipped with nonexpanding congruence of SD null strings.
\end{Lemat}
\begin{proof}
Putting $C_{ABCD} = a_{(A} b_{B} c_{C} d_{D)}$ in (\ref{warunek_rekurecyjnosci_1}) and contracting it with $a^{A}a^{B}a^{C}a^{D}$ we obtain two solutions. The first solution is $a^{A} \nabla_{m} a_{A} = 0$, so $a^{A}$ is a recurrent spinor field. Hence, it generates a nonexpanding congruence of SD null strings. The second solution gives $a^{A}b_{A}=0 \ \Longleftrightarrow \ b_{A} \sim a_{A}$. Hence, $a_{A}$ is double Penrose spinor. In this case Petrov-Penrose type of the SD Weyl spinor is $[\textrm{II}]$ so $C_{ABCD} = a_{(A} a_{B} c_{C} d_{D)}$. Feeding (\ref{warunek_rekurecyjnosci_1}) with $C_{ABCD} = a_{(A} a_{B} c_{C} d_{D)}$ and contracting it with $a^{A}a^{B}a^{C}$ one find that $a_{A}$ is a recurrent spinor field or $c_{A} \sim a_{A}$ so $a_{A}$ is triple Penrose spinor. In this case algebraic type of the SD Weyl spinor is  $[\textrm{III}]$, equivalently $C_{ABCD} = a_{(A} a_{B} a_{C} d_{D)}$. Putting this form of $C_{ABCD}$ once again in (\ref{warunek_rekurecyjnosci_1}) and contracting with $a^{A}a^{B}$ one gets that $a_{A}$ is recurrent or $d_{A} \sim a_{A}$ so $a_{A}$ is quadruple Penrose spinor. Hence, the type is $[\textrm{N}]$ and  $C_{ABCD} = a_{A} a_{B} a_{C} a_{D}$. The last step is to put $C_{ABCD} = a_{A} a_{B} a_{C} a_{D}$ in (\ref{warunek_rekurecyjnosci_1}) and contracting it with $a^{A}$. It proves, that $a_{A}$ is recurrent. 
\end{proof}

From Lemma \ref{Lemat_o_strunach} it follows that if a space is two-sided conformally recurrent and $C_{ABCD} \ne 0$ then SD Weyl spinor could be of the type $[\textrm{II}]^{n}$, $[\textrm{III}]^{n}$, $[\textrm{N}]^{n}$, $[\textrm{D}]^{n}$, $[\textrm{D}]^{ne}$ or $[\textrm{D}]^{nn}$. In fact, only types $[\textrm{N}]^{n}$ and $[\textrm{D}]^{nn}$ are admitted. Indeed, one can prove the following
\begin{Twierdzenie} 
\label{Twierdzenie_o_typie}
Let $C_{ABCD} \ne 0$ and $\nabla_{m} C_{ABCD} = r_{m} C_{ABCD}$. Then $C_{ABCD}$ is of the type $[\textrm{D}]^{nn}$ or $[\textrm{N}]^{n}$.
\end{Twierdzenie}
\begin{proof}
From Lemma \ref{Lemat_o_strunach} we conclude, that if $C_{ABCD} \ne 0$ and $\nabla_{m} C_{ABCD} = r_{m} C_{ABCD}$ then SD Weyl is algebraically degenerated and its multiple Penrose spinor is recurrent
\begin{equation}
\label{pomocnicza_postac_C}
C_{ABCD} = a_{(A} a_{B} b_{C} c_{D)}, \ \nabla_{m} a_{A} = Z_{m} a_{A}
\end{equation}
$Z_{m}$ is the Sommers vector. Putting (\ref{pomocnicza_postac_C}) in (\ref{warunek_rekurecyjnosci_1}) one finds
\begin{equation}
a_{(A} a_{B} \Sigma_{CD)}=0, \ \Sigma_{CD} := \nabla_{m} (b_{C}c_{D}) + (2Z_{m}- r_{m} ) b_{C}c_{D}  
\end{equation}
Hence, $\Sigma_{(AB)}=0$ what implies
\begin{equation}
\label{pomocnicze_rownanie}
\nabla_{m} (b_{(C}c_{D)}) + (2Z_{m}- r_{m} ) b_{(C}c_{D)}  = 0
\end{equation}
Contracting (\ref{pomocnicze_rownanie}) with $b^{C}b^{D}$ one finds two solutions: $(i)$: $b^{C} \nabla_{m} b_{C} = 0$ or $(ii)$: $b^{D}c_{D}=0$.

$(i)$. The first solution implies that $b_{D}$ is recurrent so it is multiple Penrose spinor. If $b_{D}$ is double Penrose spinor, then it is proportional to $c_{D}$. The nonexpanding congruence of SD null strings generated by the spinor $b_{D}$ is different to the congruence generated by the spinor $a_{A}$, so $\nabla_{m} b_{D} = S_{m} b_{D}$, where $S_{m}$ is Sommers vector of this congruence. Hence, $C_{ABCD} = a_{(A} a_{B} b_{C} b_{D)}$ and the type is $[\textrm{D}]^{nn}$. If $b_{C}$ is triple, then it must be proportional to $a_{A}$ so $C_{ABCD} = a_{(A} a_{B} a_{C} c_{D)}$ and Eq. (\ref{pomocnicze_rownanie}) takes the form
\begin{equation}
\label{pomocnicze_rownanie_2}
\nabla_{m} c_{D} + (3Z_{m}- r_{m} ) c_{D}  = 0
\end{equation}
Eq. (\ref{pomocnicze_rownanie_2}) implies that $c_{D}$ is recurrent so it is multiple Penrose spinor and it must be proportional to $a_{A}$. In this case we land at the type $[\textrm{N}]^{n}$.

$(ii)$. The second solution implies that $b_{D}$ is proportional to $c_{D}$, so it is double Penrose spinor. The SD Weyl spinor takes the form $C_{ABCD} = a_{(A} a_{B} b_{C} b_{D)}$ and is of the type $[\textrm{D}]^{n}$. However, if we consider (\ref{pomocnicze_rownanie}) written for $c_{D} = b_{D}$ and contracted with $b^{D}$, we find that $b_{D}$ is recurrent. Finally, the type reduces to $[\textrm{D}]^{nn}$.
\end{proof}

From Eq. (\ref{pomocnicze_rownanie}) one finds relation between vector $r_{m}$ and Sommers vectors of the congruences of null strings 
\begin{subequations}
\begin{eqnarray}
\label{wektor_r_przez_Sommersy}
\textrm{for the type } [\textrm{D}]^{nn}: && r_{m} = 2Z_{m} + 2S_{m}
\\ 
\textrm{for the type } [\textrm{N}]^{n}: && r_{m} = 4 Z_{m}
\end{eqnarray}
\end{subequations}

Theorem analogous to the Theorem \ref{Twierdzenie_o_typie} holds true for the ASD Weyl spinor. Moreover, the existence of the nonexpanding congruence of null strings implies, that for the type $[\textrm{D}]^{nn}$ the curvature scalar $R \ne 0$ and for the type $[\textrm{N}]^{n}$ the curvature scalar $R= 0$ \cite{Plebanski_Przanowski_rec,Chudecki_struny}. Hence, the spaces with mixed curvature, i.e., spaces of the types $[\textrm{D}]^{nn} \otimes [\textrm{N}]^{n}$ or $[\textrm{N}]^{n} \otimes [\textrm{D}]^{nn}$ cannot be two-sided conformally recurrent. The only possibilities are $[\textrm{N}]^{n} \otimes [\textrm{N}]^{n}$, $[\textrm{D}]^{nn} \otimes [\textrm{D}]^{nn}$, $[\textrm{N}]^{n} \otimes [-]^{n}$ and $[\textrm{D}]^{nn} \otimes [-]^{e}$ (compare Theorem \ref{Uwaga_1}). All of these types except the last one have been explicitly found in \cite{Plebanski_Przanowski_rec}. 

To find the metric of the two-sided conformally recurrent space of the type  $[\textrm{D}]^{nn} \otimes [-]^{e}$ we use the \textsl{hyperheavenly spaces formalism}. As a starting structure we consider $(\mathcal{M}, ds^{2})$ equipped with one congruence of SD null strings. Such a space is the algebraically degenerated \textsl{Walker space} or - equivalently - \textsl{weak nonexpanding hyperheavenly space}.

\section{Weak nonexpanding hyperheavenly spaces}
\setcounter{equation}{0}

\subsection{Spaces of the types $[\textrm{deg}]^{n} \otimes [\textrm{any}]$}

\subsubsection{The metric, connection and curvature}

\textsl{Weak hyperheavenly space} has been defined in \cite{Chudecki_Przanowski_Walkery}.

\begin{Definicja}
\textsl{Weak hyperheavenly space (weak $\mathcal{HH}$-space)} is a 4-dimensional complex analytic differential manifold $\mathcal{M}$ endowed with a holomorphic metric $ds^{2}$ satisfying the following conditions:
\begin{itemize}
\item there exists a 2-dimensional holomorphic totally null self-dual integrable distribution given by the Pfaff system
\begin{equation}
\label{condition_1}
m_{A} \, g^{A \dot{B}} = 0 \ \ \ \ , \ \ \ \ m_{A} \ne 0
\end{equation}
\item the self-dual Weyl spinor $C_{ABCD}$ is algebraically degenerate and $m_{A}$ is a multiple Penrose spinor i.e.
\begin{equation}
\label{condition_2}
C_{ABCD} \, m^{A} m^{B} m^{C} =0 
\end{equation}
\end{itemize}
\end{Definicja}
If we additionally assume that the congruence of SD null strings is nonexpanding then the metric of such a \textsl{nonexpanding weak $\mathcal{HH}$-space} without any loss of generality can be brought to the form  \cite{Plebanski_Rozga,Chudecki_Przanowski_Walkery} 
\begin{equation}
\label{metryka_slaba_HH}
\frac{1}{2} \, ds^{2}=  -dp^{\dot{A}}  dq_{\dot{A}} + Q^{\dot{A}\dot{B}} \, dq_{\dot{A}}  dq_{\dot{B}} 
\end{equation}
where $Q^{\dot{A}\dot{B}} = Q^{(\dot{A}\dot{B})}$ are holomorphic functions and $(q_{\dot{A}}, p_{\dot{B}})$ are complex variables. Coordinates $p_{\dot{A}}$ are coordinates on null strings while coordinates $q_{\dot{A}}$ label the null strings. Equivalently, leafs of foliation are given by $q_{\dot{A}}=\textrm{const}$. Introduce the null tetrad
\begin{eqnarray}
\label{tetrada_Plebanskiego}
&& [e^{3} , e^{1}] = -\frac{1}{\sqrt{2}} \, g^{2}_{\ \, \dot{A}} = dq_{\dot{A}} 
\\ \nonumber
&& [  e^{4}, e^{2} ]= \frac{1}{\sqrt{2}} \, g^{1 \dot{A}} =-dp^{\dot{A}} + Q^{\dot{A} \dot{B}} \, dq_{\dot{B}} 
\end{eqnarray}
Tetrad (\ref{tetrada_Plebanskiego}) is so called \textsl{Plebański tetrad}. If we define operators
\begin{eqnarray}
&& \partial_{\dot{A}} := \frac{\partial}{\partial p^{\dot{A}}}, \ \eth_{\dot{A}} :=  \left( \frac{\partial}{\partial q^{\dot{A}}} - Q_{\dot{A}}^{\ \ \dot{B}} \partial_{\dot{B}} \right)
\\ \nonumber
&& \partial^{\dot{A}} := \frac{\partial}{\partial p_{\dot{A}}}, \ 
\eth^{\dot{A}} :=  \left( \frac{\partial}{\partial q_{\dot{A}}} + Q^{\dot{A} \dot{B}} \partial_{\dot{B}} \right)
\end{eqnarray}
then the relation between null tetrads $(\partial_{1} , \partial_{2}, \partial_{3} , \partial_{4})$ and $(\partial_{\dot{A}}, \eth_{\dot{A}})$ reads
\begin{equation}
[   \partial_{4} , \partial_{2} ] = -\partial_{\dot{A}}  , \
[   \partial_{3}  ,\partial_{1} ]  = \eth^{\dot{A}} 
\end{equation}
Decomposing the connection forms according to the formulas
\begin{equation}
\label{rozklad_form_koneksji}
\mathbf{\Gamma}_{AB} = -\frac{1}{2} \mathbf{\Gamma}_{AB M \dot{N}} \, g^{M \dot{N}}
\ , \ \ \ 
\mathbf{\Gamma}_{\dot{A}\dot{B}} = -\frac{1}{2} \mathbf{\Gamma}_{\dot{A}\dot{B} M \dot{N}} \, g^{M \dot{N}}
\end{equation}
one finds that nonzero connection coefficients are 
\begin{equation}
\label{rozpisane_wspolczynniki_koneksji_spinorowej}
\mathbf{\Gamma}_{122 \dot{D}} = - \frac{1}{\sqrt{2}} \,  \partial^{\dot{A}}  Q_{\dot{A}\dot{D}}, \
\mathbf{\Gamma}_{222 \dot{D}} = -\sqrt{2}  \, \eth^{\dot{A}} Q_{\dot{A}\dot{D}}, \
\mathbf{\Gamma}_{\dot{A} \dot{B} 2 \dot{D}} =  \sqrt{2} \,  
 \partial_{(\dot{A}} Q_{\dot{B})\dot{D}}  
\end{equation}
The SD curvature coefficients $C^{(i)}$ and the curvature scalar $R$ read
\begin{eqnarray}
\label{krzywizna}
&& C^{(5)}  = C^{(4)} = 0, \ C^{(3)} = \frac{R}{6} = -\frac{1}{3} \, \partial_{\dot{A}} \partial_{\dot{B}} Q^{\dot{A} \dot{B}}, \
 C^{(2)}  = - \,  \partial^{\dot{A}}  \eth^{\dot{B}} Q_{\dot{A} \dot{B}}
\\ \nonumber
&& \frac{1}{2} \, C^{(1)}  = 
- \eth^{\dot{A}}  \eth^{\dot{B}} Q_{\dot{A} \dot{B}}  +
 (\eth^{\dot{A}} Q_{\dot{A} \dot{B}})
(\partial_{\dot{C}} Q^{\dot{B} \dot{C}}) , \ C_{\dot{A}\dot{B}\dot{C}\dot{D}} = -  
\partial_{(\dot{A}}\partial_{\dot{B}} Q_{\dot{C} \dot{D})}
\end{eqnarray}
(where $C^{(5)} \sim C_{1111}$, $C^{(4)} \sim C_{1112}$, $C^{(3)} \sim C_{1122}$,...). Finally, the components of the traceless Ricci tensor are given by the formulas
\begin{equation}
\label{Traceless_Ricci}
C_{11 \dot{A} \dot{B}} =0, \ C_{12 \dot{A} \dot{B}} = - \frac{1}{2}
 \partial_{(\dot{A}} \partial^{\dot{C}} Q_{\dot{B}) \dot{C}} , \
C_{22 \dot{A} \dot{B}} = -  \partial_{(\dot{A}}  \eth^{\dot{C}} Q_{\dot{B}) \dot{C}}
\end{equation}
The members $(\partial_{2}, \partial_{4})$ of Plebański tetrad are tangent to the null strings. In this sense Plebański tetrad is \textsl{adapted} to the congruence of SD null strings. Spinor field which generates the congruence has the form $m_{A} = (0,m)$, $m \ne 0$. Petrov-Penrose types of SD Weyl spinor are given by the conditions
\begin{eqnarray}
\label{warunki_na_typy}
\textrm{type [II]}: && C^{(3)} \ne 0, 2C^{(2)} C^{(2)} - 3 C^{(1)} C^{(3)} \ne 0
\\ \nonumber
\textrm{type [D]}: && C^{(3)} \ne 0, 2C^{(2)} C^{(2)} - 3 C^{(1)} C^{(3)} = 0
\\ \nonumber
\textrm{type [III]}: && C^{(3)} = 0, C^{(2)} \ne 0
\\ \nonumber
\textrm{type [N]}: && C^{(3)} =  C^{(2)} = 0, C^{(1)} \ne 0
\\ \nonumber
\textrm{type } [-]: && C^{(3)} =  C^{(2)} =  C^{(1)} = 0
\end{eqnarray}
[\textbf{Remark}: The metric (\ref{metryka_slaba_HH}) is double Kerr-Schild-equivalent to the flat metric \cite{Plebanski_Schild}. Therefore it belongs to the \textsl{double Kerr-Schild (dKS)} class. If
\begin{equation}
\label{warunek_na_pojedynczego_Kerra}
Q^{\dot{A}\dot{B}} Q_{\dot{A}\dot{B}}=0
\end{equation}
then the metric (\ref{metryka_slaba_HH}) exhibits the \textsl{single Kerr-Schild (sKS)} structure.]

\subsubsection{Coordinate gauge freedom}

The metric (\ref{metryka_slaba_HH}) remains invariant under the following transformations of the coordinates
\begin{equation}
\label{gauge}
q'_{\dot{A}} = q'_{\dot{A}} (q_{\dot{B}}), \ p'^{\dot{A}} = D^{-1 \ \ \dot{A}}_{\ \ \; \dot{B}} \, p^{\dot{B}} + \sigma^{\dot{A}}
\end{equation}
where $\sigma^{\dot{A}}=\sigma^{\dot{A}}(q_{\dot{B}})$ are arbitrary functions and 
\begin{eqnarray}
\nonumber
D_{\dot{A}}^{\ \ \dot{B}} 
&:=& \frac{\partial q'_{\dot{A}}}{\partial q_{\dot{B}}} 
= \Delta \, \frac{\partial q^{\dot{B}}}{\partial q'^{\dot{A}}}
= \Delta \, \frac{\partial p'_{\dot{A}}}{\partial p_{\dot{B}}}
= \frac{\partial p^{\dot{B}}}{\partial p'^{\dot{A}}} , \ \Delta := \det 
\left( \frac{\partial q'_{\dot{A}}}{\partial q_{\dot{B}}} \right)
 = \frac{1}{2} \, D_{\dot{A}\dot{B}} D^{\dot{A}\dot{B}}
\\ 
D^{-1 \ \ \dot{B}}_{\ \ \; \dot{A}} 
&=& \frac{\partial q_{\dot{A}}}{\partial q'_{\dot{B}}}
= \Delta^{-1} \, \frac{\partial q'^{\dot{B}}}{\partial q^{\dot{A}}}
= \frac{\partial p'^{\dot{B}}}{\partial p^{\dot{A}}}
= \Delta^{-1} \, \frac{\partial p_{\dot{A}}}{\partial p'_{\dot{B}}}
\end{eqnarray}
Hence
\begin{eqnarray}
 D^{\, \dot{A}}_{\ \ \dot{B}} &:=& D_{\dot{M}}^{\ \ \dot{N}} \in^{\dot{M} \dot{A}} \in_{\dot{B} \dot{N}} = - \Delta \, D^{-1 \ \ \dot{A}}_{\ \ \; \dot{B}}
\\ \nonumber
D^{-1 \; \dot{A}}_{\ \ \ \ \ \dot{B}} &:=& D^{-1 \ \ \dot{N}}_{\ \ \; \dot{M}} \in^{\dot{M} \dot{A}} \in_{\dot{B} \dot{N}} = - \frac{1}{\Delta} \, D_{\dot{B}}^{\ \ \dot{A}}
\end{eqnarray}
Functions $Q^{\dot{A}\dot{B}}$ transform under (\ref{gauge}) as follows
\begin{equation}
\label{transformacja_Q}
Q'^{\dot{A} \dot{B}} = D^{-1 \ \ \dot{A}}_{\ \ \; \dot{R}} \, D^{-1 \ \ \dot{B}}_{\ \ \; \dot{S}}
Q^{\dot{R} \dot{S}} + D^{-1 \ \ ( \dot{A}}_{\ \ \; \dot{R}} \, 
\frac{\partial p'^{\dot{B})}}{\partial q_{\dot{R}}} 
\end{equation}
Transformations (\ref{gauge}) are equivalent to the spinorial transformations 
\begin{eqnarray}
\label{transformacja_spinorowa}
L^{A}_{\ \; B} &=& \left[ \begin{array}{cc}
       \Delta^{-\frac{1}{2}} & \tilde{h}  \Delta^{\frac{1}{2}}   \\
       0 &  \Delta^{\frac{1}{2}}  
       \end{array} \right], \ 2\tilde{h} := \frac{\partial \sigma^{\dot{R}}}{\partial q'^{\dot{R}}}
\\ \nonumber
M^{\dot{A}}_{\ \; \dot{B}} &=& \Delta^{\frac{1}{2}} \, D^{-1 \ \ \dot{A}}_{\ \ \; \dot{B}}
\end{eqnarray}
Hence, dotted and undotted spinors transform according to the formulas
\begin{equation}
m'^{A} = L^{A}_{\ \; B} m^{B}, \ m'^{\dot{A}} = M^{\dot{A}}_{\ \; \dot{B}} m^{\dot{B}}
\end{equation}

\subsection{Spaces of the types $[\textrm{deg}]^{n} \otimes [-]$}

Weak $\mathcal{HH}$-spaces formalism (also called \textsl{the second Plebański's formalism}) has one major advantage. Using this formalism one can very easily pass to the algebraically degenerated nonexpanding SD spaces, i.e., spaces of the types $[\textrm{deg}]^{n} \otimes [-]$. This transition is realized by the algebraic condition $C_{\dot{A}\dot{B}\dot{C}\dot{D}}=0$ (\ref{krzywizna}) which can be immediately solved \cite{Chudecki_Przanowski_Walkery, Plebanski_Przanowski_rec}. For the spaces of the types $[\textrm{deg}]^{n} \otimes [-]$ one finds
\begin{equation}
\label{ogolne_rozwiazanie_na_Q}
Q^{\dot{A}\dot{B}} = A^{\dot{N}} p_{\dot{N}} p^{\dot{A}} p^{\dot{B}} + B^{\dot{N}(\dot{A}} p^{\dot{B})} p_{\dot{N}} + B p^{\dot{A}} p^{\dot{B}} + C^{\dot{A}\dot{B}\dot{N}} p_{\dot{N}} +C^{(\dot{A}} p^{\dot{B})} + E^{\dot{A}\dot{B}}
\end{equation}
where $A^{\dot{N}}$, $B^{\dot{N}\dot{A}} = B^{(\dot{N}\dot{A})}$, $B$, $C^{\dot{A}\dot{B}\dot{N}}=C^{(\dot{A}\dot{B}\dot{N})}$, $C^{\dot{A}}$, $E^{\dot{N}\dot{A}} = E^{(\dot{N}\dot{A})}$ are arbitrary functions of the variables $q^{\dot{A}}$. Feeding (\ref{krzywizna}) and (\ref{Traceless_Ricci}) with (\ref{ogolne_rozwiazanie_na_Q}) one gets
\begin{eqnarray}
\label{postac_R}
&& \frac{R}{6} = C^{(3)} = 4 A_{\dot{N}} p^{\dot{N}} - 2B , \ C^{(2)}  = - \,  \partial^{\dot{A}}  Q_{\dot{A}}, \ 
 \frac{1}{2} C^{(1)}  = - \eth^{\dot{A}}  Q_{\dot{A}}  + Q_{\dot{B}} \, \partial_{\dot{C}} Q^{\dot{B} \dot{C}} \ \ \ \ \
\\ \nonumber
&& C_{12\dot{A}\dot{B}} = 2 A_{(\dot{A}} p_{\dot{B})} + B_{\dot{A}\dot{B}}, \ 
C_{22 \dot{A} \dot{B}} = -  \partial_{(\dot{A}}   Q_{\dot{B})}
\end{eqnarray}
where
\begin{eqnarray}
\label{postac_eth_Q}
Q_{\dot{B}} &:=& \eth^{\dot{A}} Q_{\dot{A}\dot{B}} 
\\ \nonumber
&=& \frac{\partial}{\partial q_{\dot{A}}} \left( A^{\dot{N}} p_{\dot{N}} p_{\dot{A}} p_{\dot{B}} + B^{\dot{N}}_{\ (\dot{A}} p_{\dot{B})} p_{\dot{N}} + B p_{\dot{A}} p_{\dot{B}} + C^{\dot{N}}_{\ \dot{A}\dot{B}} p_{\dot{N}} +C_{(\dot{A}} p_{\dot{B})} + E_{\dot{A}\dot{B}} \right) \ \ \ \ \ \ \ 
\\ \nonumber
&&-A_{\dot{C}}C^{\dot{C}\dot{A}\dot{X}}p_{\dot{X}}p_{\dot{A}}p_{\dot{B}} + \frac{1}{2} A^{\dot{C}} C^{\dot{A}} p_{\dot{C}}p_{\dot{A}}p_{\dot{B}} - \frac{1}{2} BB^{\dot{A}\dot{C}} p_{\dot{C}}p_{\dot{A}}p_{\dot{B}}
\\ \nonumber
&& -B_{\dot{B}\dot{C}} C^{\dot{A}\dot{C}\dot{X}} p_{\dot{A}} p_{\dot{X}} - A_{\dot{C}} E^{\dot{A}\dot{C}} p_{\dot{A}} p_{\dot{B}} + A^{\dot{N}} E^{\dot{A}}_{\ \dot{B}} p_{\dot{A}} p_{\dot{N}} - \frac{3}{4} C_{\dot{X}} B^{\dot{X}\dot{Z}} p_{\dot{Z}} p_{\dot{B}}
\\ \nonumber
&&-\frac{1}{4}C^{\dot{A}} B^{\dot{X}}_{\ \dot{B}} p_{\dot{A}}p_{\dot{X}} - p_{\dot{X}} \left( C^{\dot{A}\dot{C}\dot{X}}C_{\dot{A}\dot{C}\dot{B}} + \frac{1}{4}C_{\dot{B}}C^{\dot{X}} - BE^{\dot{X}}_{\ \dot{B}}   \right)
\\ \nonumber
&& -  E^{\dot{A}\dot{C}}B_{\dot{A}\dot{C}} p_{\dot{B}}  - \frac{1}{2} C_{\dot{B}}^{\ \dot{A}\dot{C}} C_{\dot{A}} p_{\dot{C}} -E^{\dot{A}\dot{C}} C_{\dot{A}\dot{C}\dot{B}} - \frac{1}{2} C^{\dot{A}} E_{\dot{A}\dot{B}}
\end{eqnarray}
After tedious but straightforward calculations one arrives at the following transformation formulas
\begin{subequations}
\label{wzory_transformacyjne_ogolne}
\begin{eqnarray}
 A'^{\dot{N}} &=& \Delta D^{-1 \ \ \dot{N}}_{\ \ \; \dot{M}} A^{\dot{M}}
\\
B' &=& B + 2 D_{\dot{N}}^{\ \ \dot{M}} \sigma^{\dot{N}} A_{\dot{M}}
\\
 \Delta B'_{\dot{N}\dot{X}} &=& D_{\dot{X}}^{\ \ \dot{B}} D_{\dot{N}}^{\ \ \dot{R}} B_{\dot{R}\dot{B}} - 2 \Delta A_{\dot{R}} \, \sigma_{(\dot{X}} D_{\dot{N})}^{\ \ \dot{R}}
\\ 
\label{transformacja_na_CA}
 C'^{\dot{B}} &=& D^{-1 \ \ \dot{B}}_{\ \ \; \dot{S}} C^{\dot{S}} - \frac{2}{3} \frac{\partial \ln \Delta}{\partial q'_{\dot{B}}} - 2 B' \sigma^{\dot{B}} - \frac{4}{3} \sigma_{\dot{N}} B'^{\dot{N}\dot{B}} + \frac{8}{3} A'_{\dot{N}} \sigma^{\dot{N}} \sigma^{\dot{B}}
\\
\label{transformacja_na_CABR}
 C'^{\dot{A}\dot{B}\dot{R}} &=& \Delta D^{-1 \ \ \dot{A}}_{\ \ \; \dot{X}} D^{-1 \ \ \dot{B}}_{\ \ \; \dot{S}} D^{-1 \ \ \dot{R}}_{\ \ \; \dot{M}} C^{\dot{X}\dot{S}\dot{M}} +  D^{-1 \ \ (\dot{A}}_{\ \ \; \dot{S}}  D^{-1 \ \ \dot{B}}_{\ \ \; \dot{M}} \frac{\partial^{2} q'^{\dot{R})}}{\partial q_{\dot{S}} \partial q_{\dot{M}}}
\\ \nonumber
&& -A'^{(\dot{A}} \sigma^{\dot{B}} \sigma^{\dot{R})} - B'^{(\dot{A}\dot{B}} \sigma^{\dot{R})}
\\
E'^{\dot{A}\dot{B}} &=& D^{-1 \ \ \dot{A}}_{\ \ \; \dot{R}} D^{-1 \ \ \dot{B}}_{\ \ \; \dot{S}} E^{\dot{R}\dot{S}}
+ D^{-1 \ \ (\dot{A}}_{\ \ \; \dot{R}} \frac{\partial \sigma^{\dot{B})}}{\partial q_{\dot{R}}}
-C'^{(\dot{A}}\sigma^{\dot{B})} -B' \sigma^{\dot{A}}\sigma^{\dot{B}}
\\ \nonumber
&&+ A'_{\dot{N}} \sigma^{\dot{N}} \sigma^{\dot{A}} \sigma^{\dot{B}} - C'^{\dot{N}\dot{A}\dot{B}} \sigma_{\dot{N}} - B'^{\dot{N}(\dot{A}} \sigma^{\dot{B})} \sigma_{\dot{N}}
\end{eqnarray}
\end{subequations}

\section{Spaces of the type $[\textrm{D}]^{nn} \otimes [-]^{e}$}
\setcounter{equation}{0}

\subsection{Two-sided conformally recurrent spaces of the type $[\textrm{D}]^{nn} \otimes [-]^{e}$}

Except obvious advantage (compare the solution (\ref{ogolne_rozwiazanie_na_Q})), the $\mathcal{HH}$-spaces formalism has one major disadvantage. Namely, it is quite hard to pass from the spaces of the types $[\textrm{deg}]^{n} \otimes [-]$ to the spaces of the type $[\textrm{D}]^{nn} \otimes [-]^{e}$. Indeed, the condition for $C_{ABCD}$ to be of the type $[\textrm{D}]$ reads
\begin{equation}
\label{warunek_na_typp_D}
2C^{(2)} C^{(2)} - 3 C^{(1)} C^{(3)} = 0
\end{equation}
Both $C^{(1)}$ and $C^{(2)}$ depends on $\eth^{\dot{A}} Q_{\dot{A}\dot{B}}$ (\ref{postac_eth_Q}) and its derivatives. Clearly, (\ref{warunek_na_typp_D}) becomes very complicated algebraic condition. 

[\textbf{Remark}: $\mathcal{HH}$-spaces formalism is not the only formalism which can be used to obtain metrics of the SD spaces of the type [D] which are two-sided conformally recurrent. The alternative approach uses the null tetrad which is adapted to the two distinct congruences of SD null strings. If both these congruences are nonexpanding then the space is of the type $[\textrm{D}]^{nn} \otimes [\textrm{any}]$ and the metric has the form
\begin{equation}
\frac{1}{2} ds^{2} = \frac{\partial^{2} \mathcal{F}}{\partial q^{A} \partial q^{\dot{B}}} \, dq^{A} dq^{\dot{B}}
\end{equation}
Coordinates $(q^{A},q^{\dot{A}})$ are so called \textsl{double null coordinates}. Although double null coordinates system seems to be perfectly adapted to the spaces equipped with two distinct nonexpanding congruences of SD null strings, it has one severe disadvantage. To pass to the SD spaces of the type $[\textrm{D}]^{nn} \otimes [-]$ one has to solve the condition $C_{\dot{A}\dot{B}\dot{C}\dot{D}}=0$. It is equivalent to the set of five nasty differential equations. This set of equations have been attacked in \cite{Przanowski_Formanski_Chudecki}, with no success.]

However, Eq. (\ref{warunek_na_typp_D}) leads to the type $[\textrm{D}]^{n} \otimes [-]^{e}$. There is no second congruence of SD null strings in general. To find the spaces which are two-sided conformally recurrent we need much more special type $[\textrm{D}]^{nn} \otimes [-]^{e}$. Instead of the condition (\ref{warunek_na_typp_D}) we focus on the equations for the second congruence of SD null strings. 

The first congruence of SD null strings is generated by the spinor $m_{A} = (0,m)$, $m\ne0$. The second congruence is distinct so it must be generated by the spinor $l_{A} = (a,b)$ such that $l^{A}m_{A} = -am \ne 0$. Hence, $a\ne0$. The null strings equations for this second congruence take the form 
\begin{equation}
l^{A} l^{B} \nabla_{A \dot{M}} l_{B} = 0
\end{equation}
and
\begin{eqnarray}
\nonumber
\textrm{for the type } [\textrm{D}]^{ne} \otimes [-]^{e}: && l^{B} \nabla_{A \dot{M}} l_{B} \ne 0   
\\ \nonumber
\textrm{for the type } [\textrm{D}]^{nn} \otimes [-]^{e}: && l^{B} \nabla_{A \dot{M}} l_{B} = 0
\end{eqnarray}
We have to solve the equations
\begin{equation}
\label{warunek_na_typp_D_nn}
l^{B} \nabla_{A \dot{M}} l_{B} =  l^{B} (\partial_{A \dot{M}} l_{B}  -  \mathbf{\Gamma}^{S}_{\ BA \dot{M}} l_{S} )   = 0
\end{equation}
where
\begin{equation}
 \partial_{A \dot{C}} := \sqrt{2} [\partial_{\dot{C}}, \eth_{\dot{C}}]
\end{equation}
Eqs. (\ref{warunek_na_typp_D_nn}) written for $A=1$ give
\begin{equation}
l^{B} \frac{\partial l_{B}}{\partial p^{\dot{M}}} = 0 \ \Longleftrightarrow \ b \frac{\partial a}{\partial p^{\dot{M}}} - a \frac{\partial b}{\partial p^{\dot{M}}} = 0
\end{equation}
with the solution $b = \tilde{f}(q^{\dot{M}}) a$. Transformation for the function $\tilde{f}$ reads 
\begin{equation}
\label{transformacja_na_funckje_f}
\tilde{f}'=\Delta^{-1}\tilde{f}-\tilde{h}
\end{equation}
Using gauge function $\tilde{h}$ one can put $\tilde{f}=0$ without any loss of generality. Hence, $l_{A}=(a,0)$, $a \ne 0$. Eqs. (\ref{warunek_na_typp_D_nn}) written for $A=2$ give
\begin{equation}
0= a^{2} \mathbf{\Gamma}_{222\dot{M}} \ \Longrightarrow \ \eth^{\dot{A}} Q_{\dot{A}\dot{B}}=0
\end{equation}
Condition $\eth^{\dot{A}} Q_{\dot{A}\dot{B}}=0$ implies $C^{(2)}=C^{(1)}=C_{22\dot{A}\dot{B}}=0$. $\eth^{\dot{A}} Q_{\dot{A}\dot{B}}=0$ becomes a set of two third-order polynomials in $p^{\dot{A}}$ with the coefficients which depend on $q^{\dot{A}}$ only. Explicitly, we arrive at the following system of 15 equations for 15 functions
\begin{subequations}
\label{koncowy_uklad}
\begin{eqnarray}
\label{koncowy_uklad_1}
&&  \frac{\partial A_{(\dot{A}}}{\partial q^{\dot{B})}} - C_{\dot{A}\dot{B}\dot{X}}A^{\dot{X}} - \frac{1}{2} C_{(\dot{A}} A_{\dot{B})} + \frac{1}{2} B B_{\dot{A}\dot{B}} = 0
\\ 
\label{koncowy_uklad_2}
&& \frac{\partial E_{\dot{A}\dot{B}}}{\partial q_{\dot{A}}} - E^{\dot{A}\dot{C}}C_{\dot{A}\dot{C}\dot{B}} - \frac{1}{2}C^{\dot{A}}E_{\dot{A}\dot{B}} =0
\\
\label{koncowy_uklad_3}
&&\frac{3}{2} \frac{\partial B}{\partial q^{\dot{A}}} - \frac{\partial B_{\dot{A}\dot{N}}}{\partial q_{\dot{N}}} + C_{\dot{A}\dot{X}\dot{B}}B^{\dot{X}\dot{B}} - 2A^{\dot{B}}E_{\dot{B}\dot{A}} - C^{\dot{X}}B_{\dot{X}\dot{A}} = 0
\\
\label{koncowy_uklad_4}
&&  \frac{1}{2} \frac{B_{(\dot{A}\dot{B}}}{\partial q^{\dot{C})}} - B^{\dot{X}}_{\ (\dot{A}} C_{\dot{B}\dot{C})\dot{X}} - A_{(\dot{A}} E_{\dot{B}\dot{C})} + \frac{1}{4} C_{(\dot{A}} B_{\dot{B}\dot{C})}=0
\\
\label{koncowy_uklad_5}
&&\frac{3}{2} \frac{\partial C^{\dot{A}}}{\partial q^{\dot{A}}} - 2B^{\dot{A}\dot{B}}E_{\dot{A}\dot{B}} = 0
\\ 
\label{koncowy_uklad_6}
&& \frac{\partial C_{\dot{A}\dot{B}\dot{C}}}{\partial q_{\dot{A}}} - \frac{1}{2} \frac{\partial C_{(\dot{B}}}{\partial q^{\dot{C})}} -C_{\dot{A}\dot{X}\dot{B}}C^{\dot{A}\dot{X}}_{\ \ \ \dot{C}} - \frac{1}{4} C_{\dot{B}} C_{\dot{C}} + BE_{\dot{B}\dot{C}} + \frac{1}{2} C^{\dot{A}} C_{\dot{A}\dot{B}\dot{C}}=0
\end{eqnarray}
\end{subequations}
Eqs. (\ref{koncowy_uklad}) look nasty, but, perhaps surprisingly, they can be completely solved in all generality. From (\ref{postac_R}) it follows, that $A_{\dot{N}}$ and $B$ cannot vanish simultaneously. Otherwise, the curvature scalar $R=0$ and the SD Weyl spinor is not of the type [D] anymore. Thus we have
\begin{Lemat}
\label{Lemat_o_przypadkach}
For the spaces of the type $[\textrm{D}]^{nn} \otimes [-]^{e}$ the following statements are equivalent
\begin{eqnarray}
\nonumber
(i) && A_{\dot{N}} = 0
\\ \nonumber
(ii) && C_{12\dot{A}\dot{B}} = 0 
\\ \nonumber
(iii) && R = \textrm{const}
\end{eqnarray}
\end{Lemat}
\begin{proof}
If $A_{\dot{N}}=0$ then from (\ref{koncowy_uklad_1}) it follows $B_{\dot{A}\dot{B}}=0$, because $A_{\dot{N}}$ and $B$ cannot vanish simultaneously. It implies $C_{12\dot{A}\dot{B}} = 0$. Conversely, $C_{12\dot{A}\dot{B}} = 0$ implies $A_{\dot{N}}=0$. Thus, equivalence $(i) \Longleftrightarrow (ii)$ is completed. It remains to prove $(i) \Longleftrightarrow (iii)$. If $R=\textrm{const}$ then $A_{\dot{N}}=0$ (compare (\ref{postac_R})). If $A_{\dot{N}}=0$ then (\ref{koncowy_uklad_1}) implies $B_{\dot{A}\dot{B}}=0$ and (\ref{koncowy_uklad_3}) gives $B=\textrm{const}$; consequently $R=\textrm{const}$. It completes the proof. 
\end{proof}

From Lemma \ref{Lemat_o_przypadkach} it follows that there are two cases to be considered. The first is $A_{\dot{N}} \ne 0$ and in this case traceless Ricci tensor is nonzero. The second is $A_{\dot{N}}=0$ and in this case the space is Einstein space.

In the end of this section we write down the form of the Sommers vectors of both congruences
\begin{eqnarray}
\label{postac_wektorow_Sommersa}
&& Z_{A\dot{B}} = \sqrt{2} \left[ \partial_{\dot{M}} \ln m, \eth_{\dot{M}} \ln m + \frac{1}{2} \partial^{\dot{A}} Q_{\dot{A}\dot{M}} \right] \ne 0
\\ \nonumber
&& S_{A\dot{B}} = \sqrt{2} \left[ \partial_{\dot{M}} \ln a, \eth_{\dot{M}} \ln a - \frac{1}{2} \partial^{\dot{A}} Q_{\dot{A}\dot{M}} \right] \ne 0
\end{eqnarray}
Both these vectors are necessarily nonzero. Otherwise, the space is equipped with covariantly constant field of 1-index undotted spinors. It is possible only for the type $[\textrm{N}]^{n}$. Using (\ref{postac_wektorow_Sommersa}) one finds (compare {\ref{wektor_r_przez_Sommersy}})
\begin{equation}
r_{A\dot{B}} = 2 Z_{A\dot{B}} + 2 S_{A\dot{B}} = 2 \sqrt{2} [\partial_{\dot{M}} \ln (ma), \eth_{\dot{M}} \ln (ma)]= \partial_{A\dot{M}} \ln (am)^{2} = \partial_{A\dot{M}} \ln R
\end{equation}
Clearly, the condition $\nabla_{m} C_{ABCD} = r_{m} C_{ABCD}$ implies $\nabla_{m} R = r_{m} R$. However, simple calculations prove that the condition $\nabla_{m} C_{AB\dot{C}\dot{D}} = r_{m} C_{AB\dot{C}\dot{D}}$ cannot be satisfied for nonzero $r_{m}$, i.e., for $R \ne \textrm{const}$. Hence, any SD space of the type $[\textrm{D}]^{nn} \otimes [-]^{e}$ is conformally recurrent, but
\begin{eqnarray}
\nonumber
(i) && \textrm{if } C_{ab} = 0 \textrm{ then it is conformally symmetric and symmetric}
\\ \nonumber
(ii) && \textrm{if } C_{ab} \ne 0 \textrm{ then it cannot be conformally symmetric, nor recurrent}
\end{eqnarray}

\subsection{Non-Einstein spaces of the type $[\textrm{D}]^{nn} \otimes [-]^{e}$}

The analysis presented in this section and in the section \ref{sekcja_przestrzenie_einsteinowskie} uses the gauge freedom (\ref{gauge}). Formulas (\ref{wzory_transformacyjne_ogolne}) remain valid but $\sigma^{\dot{A}}$ are not arbitrary anymore. The ambiguity in the gauge function $\tilde{h}$ has been already used (compare (\ref{transformacja_na_funckje_f})) and from now on $\tilde{h}=0$. Hence, $\sigma^{\dot{A}} = \dfrac{\partial \sigma}{\partial q'_{\dot{A}}}$, where $\sigma=\sigma (q^{\dot{M}})$.

Because $A_{\dot{N}} \ne 0$ one can put $A_{\dot{N}} = \delta_{\dot{N}}^{\dot{1}}$ and $B=0$ without any loss of generality. Gauge transformations are reduced to
\begin{eqnarray}
\nonumber
&& q'^{\dot{1}} = q'^{\dot{1}} (q^{\dot{1}}), \ q'^{\dot{2}} = q^{\dot{2}} + \tilde{s} (q^{\dot{1}}) \ \Longrightarrow \ \Delta = \dfrac{d q'^{\dot{1}}}{d q^{\dot{1}}}
\\ \nonumber
&& \sigma=\sigma (q^{\dot{1}}) \ \Longrightarrow \ \sigma^{\dot{1}}=0 ,\ \sigma^{\dot{2}} = -\Delta^{-1} \dfrac{d \sigma}{d q^{\dot{1}}}
\end{eqnarray}
where $\tilde{s}$ and $\sigma$ are arbitrary functions of $q^{\dot{1}}$ only. If, for convenience, we denote 
\begin{equation}
\label{skroty_funkcji}
M := C_{\dot{1}\dot{1}\dot{1}}  ,\ N := C_{\dot{1}\dot{1}\dot{2}}  ,\  P := C_{\dot{1}\dot{2}\dot{2}}  ,\ S := C_{\dot{2}\dot{2}\dot{2}} 
\end{equation}
from Eqs. (\ref{koncowy_uklad_1}) one finds
\begin{equation}
 S=0, \ C_{\dot{1}} = -2N, \ C_{\dot{2}} = -4P
\end{equation}
Eqs. (\ref{koncowy_uklad_3}), (\ref{koncowy_uklad_4}) and consistency conditions between them implies
\begin{eqnarray}
\label{pierwsze_wnioski}
&& B_{\dot{A}\dot{C}} = B_{\dot{A}\dot{C}} (q^{\dot{1}}), E_{\dot{1}\dot{1}} = \frac{1}{2} \frac{d B_{\dot{1}\dot{1}}}{d q^{\dot{1}}} + MB_{\dot{1}\dot{2}} - \frac{3}{2} N B_{\dot{1}\dot{1}}
\\ \nonumber
&& E_{\dot{1}\dot{2}} = \frac{1}{2} \frac{d B_{\dot{1}\dot{2}}}{d q^{\dot{1}}} + \frac{1}{2} MB_{\dot{2}\dot{2}} - \frac{3}{2} P B_{\dot{1}\dot{1}}, \ E_{\dot{2}\dot{2}} = \frac{1}{2} \frac{d B_{\dot{2}\dot{2}}}{d q^{\dot{1}}} + \frac{3}{2} NB_{\dot{2}\dot{2}} - 3P B_{\dot{1}\dot{2}}
\end{eqnarray}
Feeding (\ref{koncowy_uklad_6}) with (\ref{pierwsze_wnioski}) yields
\begin{equation}
M=M(q^{\dot{1}}), \ P=P(q^{\dot{1}}), \ N=N(q^{\dot{1}})
\end{equation}
Eqs. (\ref{koncowy_uklad_2}) give
\begin{equation}
\label{rownanie_ostateczne_1}
E_{\dot{2}\dot{2}} = \textrm{const}_{1}, \ 
\frac{d^{2} B_{\dot{1}\dot{2}}}{d q^{\dot{1}} d q^{\dot{1}}} + 2M \frac{d B_{\dot{2}\dot{2}}}{d q^{\dot{1}}} -  3N \frac{d B_{\dot{1}\dot{2}}}{d q^{\dot{1}}} - 3 B_{\dot{1}\dot{1}} \frac{d P}{d q^{\dot{1}}} + B_{\dot{2}\dot{2}} \frac{d M}{d q^{\dot{1}}} =0
\end{equation}
the last Eq. (\ref{koncowy_uklad_5}) reduces to the form
\begin{equation}
\label{rownanie_ostateczne_2}
6P + B^{2}_{\dot{1}\dot{2}} -B_{\dot{1}\dot{1}} B_{\dot{2}\dot{2}} = \textrm{const}_{2}
\end{equation}

Finally, we have to solve three equations (\ref{rownanie_ostateczne_1}) and (\ref{rownanie_ostateczne_2}) on six functions $B_{\dot{A}\dot{C}}, M,N,P$ of one variable $q^{\dot{1}}$. There are still three arbitrary gauge functions of one variable to our disposal. Detailed analysis of the transformation formula for $B_{\dot{A}\dot{B}}$ leads to three cases
\begin{eqnarray}
\nonumber
(i) &&  B_{\dot{A}\dot{C}} B^{\dot{A}\dot{C}} = 0, \ B_{\dot{A}\dot{C}} \ne 0
\\ \nonumber
(ii) &&  B_{\dot{A}\dot{C}} B^{\dot{A}\dot{C}} \ne 0
\\ \nonumber
(iii) && B_{\dot{A}\dot{C}}  = 0 
\end{eqnarray}
Case $(i)$ gives the solution
\begin{eqnarray}
\nonumber
 && A_{\dot{1}} = 1, \ A_{\dot{2}} = 0, \ B=0, \ B_{\dot{1}\dot{1}} = B_{\dot{1}\dot{2}} = 0, \ B_{\dot{2}\dot{2}} = 1, 
\\ \nonumber
&& C_{\dot{1}\dot{1}\dot{1}}=M_{0}, \ C_{\dot{1}\dot{1}\dot{2}}=N_{0}, \ C_{\dot{1}\dot{2}\dot{2}}=P_{0}, \ C_{\dot{2}\dot{2}\dot{2}}=0, 
\\ \nonumber
&& E_{\dot{1}\dot{1}}=0,\ E_{\dot{1}\dot{2}} = \frac{1}{2} M_{0}, \ E_{\dot{2}\dot{2}} = \frac{3}{2} N_{0}, \ C_{\dot{1}} = -2N_{0} , \ C_{\dot{2}} = -4P_{0},
\\ \nonumber
\end{eqnarray}
What implies
\begin{eqnarray}
Q^{\dot{1}\dot{1}} &=& -{p^{\dot{1}}}^{3} + p^{\dot{1}} p^{\dot{2}} + 3P_{0}p^{\dot{1}} + \frac{3}{2} N_{0}
\\ \nonumber
Q^{\dot{1}\dot{2}} &=& -{p^{\dot{1}}}^{2}p^{\dot{2}} + \frac{1}{2}{p^{\dot{2}}}^{2} + 3P_{0}p^{\dot{2}} - \frac{1}{2} M_{0}
\\ \nonumber
Q^{\dot{2}\dot{2}} &=& -p^{\dot{1}}{p^{\dot{2}}}^{2} - M_{0}p^{\dot{1}} -3 N_{0} p^{\dot{2}}
\end{eqnarray}
where $M_{0}$, $N_{0}$ and $P_{0}$ are arbitrary constants. The metric reads
\begin{eqnarray}
\label{postac_dKS}
\frac{1}{2} ds^{2} &=& dydq - dxdp - (xy^{2} +M_{0} x +3N_{0} y) \, dq^{2} 
\\ \nonumber
&& + \left( 2x^{2}y - y^{2} -6P_{0}y +  M_{0}  \right) \, dq dp
+  \left( -x^{3} + x y + 3P_{0}x + \frac{3}{2} N_{0}  \right) \, dp^{2}
\end{eqnarray}
where we denoted
\begin{equation}
\label{oznaczenia_wspolrzednych}
p^{\dot{1}} =:x, \ p^{\dot{2}} =: y, \ q^{\dot{1}} =:q, \ q^{\dot{2}} =: p
\end{equation}
In this case the condition (\ref{warunek_na_pojedynczego_Kerra}) cannot be satisfied so the metric belongs to the dKS class.

Cases $(ii)$ and $(iii)$ lead to the solution
\begin{eqnarray}
\nonumber
 && A_{\dot{1}} = 1, \ A_{\dot{2}} = 0, \ B=0, \ B_{\dot{1}\dot{1}} = B_{\dot{2}\dot{2}} = 0, \ B_{\dot{1}\dot{2}} = B_{0}, 
\\ \nonumber
&& C_{\dot{1}\dot{1}\dot{1}}= C_{\dot{1}\dot{1}\dot{2}}=C_{\dot{2}\dot{2}\dot{2}}=0, \ C_{\dot{1}\dot{2}\dot{2}}=P_{0},  
\\ \nonumber
&& E_{\dot{1}\dot{1}}= E_{\dot{1}\dot{2}} = 0, \ E_{\dot{2}\dot{2}} = -3P_{0} B_{0}, \ C_{\dot{1}} = 0 , \ C_{\dot{2}} = -4P_{0}
\end{eqnarray}
Hence
\begin{eqnarray}
Q^{\dot{1}\dot{1}} &=& -{p^{\dot{1}}}^{3} + B_{0} {p^{\dot{1}}}^{2} + 3P_{0}p^{\dot{1}} -3 P_{0} B_{0}
\\ \nonumber
Q^{\dot{1}\dot{2}} &=& -{p^{\dot{1}}}^{2}p^{\dot{2}} +  3P_{0}p^{\dot{2}} 
\\ \nonumber
Q^{\dot{2}\dot{2}} &=& -p^{\dot{1}}{p^{\dot{2}}}^{2} - B_{0} {p^{\dot{2}}}^{2}
\end{eqnarray}
where $P_{0}$ and $B_{0}$ are arbitrary constants. Finally, we arrive at the metric
\begin{eqnarray}
\label{Non_Einstein_second_solution}
\frac{1}{2} ds^{2} &=& dydq - dxdp - y^{2} (x + B_{0} ) \, dq^{2} 
\\ \nonumber
&& + 2y \left(x^{2} -  3P_{0}   \right) \, dq dp
+  (x-B_{0})(3P_{0}-x^{2}) \, dp^{2}
\end{eqnarray}
In general the metric (\ref{Non_Einstein_second_solution}) belongs to the dKS class but if $3P_{0} = B_{0}^{2}$ then it becomes sKS metric.

[\textbf{Remark}. In both cases (\ref{postac_dKS}) and (\ref{Non_Einstein_second_solution}) the traceless Ricci tensor can be written in the form 
\begin{equation}
\nonumber
C_{AB \dot{C} \dot{D}} = f_{AB} f_{\dot{C}\dot{D}} \ \ \textrm{where} \ \ f_{AB} f^{AB} \ne 0, \ f_{\dot{C}\dot{D}} f^{\dot{C}\dot{D}} \ne 0
\end{equation}
so it has the form of a \textsl{general electromagnetic field} \cite{Przanowski_Plebanski_materia}. Hence, the algebraic type of the traceless Ricci tensor is $[2N_{1}-2N]_{2}$.]

\subsection{Einstein spaces of the type $[\textrm{D}]^{nn} \otimes [-]^{e}$}
\label{sekcja_przestrzenie_einsteinowskie}

In this section we deal with Einstein case which is characterized by $A_{\dot{N}}=0$. Since $C_{ab}=0$ and $R=-4\Lambda$ where $\Lambda$ is cosmological constant, one finds (compare (\ref{postac_R}))
\begin{equation}
\label{wzor_na_B0}
B=B_{0} = \frac{\Lambda}{3} = \textrm{const} \ne 0
\end{equation}
Eqs. (\ref{koncowy_uklad_1}) imply $B_{\dot{A}\dot{C}}=0$. Eqs. (\ref{koncowy_uklad_3}) and (\ref{koncowy_uklad_4}) are identically satisfied. The existence of $C=C(q^{\dot{M}})$ such that $C^{\dot{A}} = \dfrac{\partial C}{\partial q_{\dot{A}}}$ follows from Eq. (\ref{koncowy_uklad_5}). Then $C$ can be gauged away (compare (\ref{transformacja_na_CA})). Hence, $C^{\dot{A}}=0$ but from now on the gauge function $\sigma$ is not arbitrary anymore. Indeed, it reads $ \Lambda \sigma + \ln \Delta = \textrm{const}$. We are left with five equations
\begin{subequations}
\label{koncowy_uklad_einstein}
\begin{eqnarray}
\label{koncowy_uklad_2_einstein}
&& \frac{\partial E_{\dot{A}\dot{B}}}{\partial q_{\dot{A}}} - E^{\dot{A}\dot{C}}C_{\dot{A}\dot{C}\dot{B}} =0
\\
\label{koncowy_uklad_6_einstein}
&& \frac{\partial C_{\dot{A}\dot{B}\dot{C}}}{\partial q_{\dot{A}}}  -C_{\dot{A}\dot{X}\dot{B}}C^{\dot{A}\dot{X}}_{\ \ \ \dot{C}}  + B_{0} E_{\dot{B}\dot{C}} =0
\end{eqnarray}
\end{subequations}
Because $B_{0} \ne 0$, from Eqs. (\ref{koncowy_uklad_6_einstein}) one gets solutions for $E_{\dot{A}\dot{B}}$
\begin{eqnarray}
B_{0} E_{\dot{1}\dot{1}} &=& 2MP - 2N^{2} - \frac{\partial M}{\partial q^{\dot{2}}} + \frac{\partial N}{\partial q^{\dot{1}}}
\\  \nonumber
B_{0} E_{\dot{1}\dot{2}} &=& MS - NP - \frac{\partial N}{\partial q^{\dot{2}}} + \frac{\partial P}{\partial q^{\dot{1}}}
\\  \nonumber
B_{0} E_{\dot{2}\dot{2}} &=& 2NS - 2P^{2} - \frac{\partial P}{\partial q^{\dot{2}}} + \frac{\partial S}{\partial q^{\dot{1}}}
\end{eqnarray}
(where we used the abbreviations (\ref{skroty_funkcji})). Eqs. (\ref{koncowy_uklad_2_einstein}) written explicitly read
\begin{subequations}
\label{rownania_ostateczne_einsteinowskie}
\begin{eqnarray}
&& \frac{\partial}{\partial q^{\dot{2}}} \left( 3MP - 3N^{2} -\frac{\partial M}{\partial q^{\dot{2}}} + \frac{\partial N}{\partial q^{\dot{1}}} \right) - \frac{\partial}{\partial q^{\dot{1}}} \left( MS - \frac{\partial N}{\partial q^{\dot{2}}} + \frac{\partial P}{\partial q^{\dot{1}}} \right)  
\\ \nonumber
&&  \ \ \ \ \ \ \ \ \ \ \ \ \ \ \ \ \ \ \ 
+ 3N \frac{\partial P}{\partial q^{\dot{1}}} - M \frac{\partial S}{\partial q^{\dot{1}}} = 0
\\ 
&& \frac{\partial}{\partial q^{\dot{2}}} \left( MS  -\frac{\partial N}{\partial q^{\dot{2}}} + \frac{\partial P}{\partial q^{\dot{1}}} \right) - \frac{\partial}{\partial q^{\dot{1}}} \left( 3NS-3P^{2} - \frac{\partial P}{\partial q^{\dot{2}}} + \frac{\partial S}{\partial q^{\dot{1}}} \right)  
\\ \nonumber
&& \ \ \ \ \ \ \ \ \ \ \ \ \ \ \ \ \ \ \ 
-3P \frac{\partial N}{\partial q^{\dot{2}}} +S\frac{\partial M}{\partial q^{\dot{2}}} = 0
\end{eqnarray}
\end{subequations}
At this point the crucial step is careful analysis of the transformation formulas (\ref{transformacja_na_CABR}). One finds
\begin{eqnarray}
\Delta^{2} M' &=& \left( \frac{\partial q'^{\dot{2}}}{\partial q^{\dot{2}}} \right)^{3} \left(  M - 3 \tau N + 3 \tau ^{2} P - \tau^{3} S + \frac{\partial \tau}{\partial q^{\dot{1}}} - \tau \frac{\partial \tau}{\partial q^{\dot{2}}} \right) 
\\ \nonumber
&& \tau := \frac{\partial q'^{\dot{2}}}{\partial q^{\dot{1}}} \left( \frac{\partial q'^{\dot{2}}}{\partial q^{\dot{2}}} \right)^{-1}
\end{eqnarray}
Using the ambiguity in the gauge function $q'^{\dot{2}} = q'^{\dot{2}} (q^{\dot{M}})$ (consequently, in $\tau$) the function $M$ can be gauged away. To keep $M=0$ gauge freedom is limited to the transformations $q'^{\dot{2}} = q'^{\dot{2}} (q^{\dot{2}})$, i.e., $\tau=0$\footnote{To be more precise, the choice $M=0$ can be maintained with nonzero $\tau$. However, we put $\tau=0$ for clarity of further investigations.}. Transformation for $N$ reads now
\begin{equation}
\Delta^{2} N' = \left( \frac{dq'^{\dot{2}}}{d q^{\dot{2}}} \right)^{2} \left( \frac{\partial q'^{\dot{1}}}{\partial q^{\dot{1}}} N - \frac{1}{3} \frac{\partial^{2} q'^{\dot{1}}}{\partial q^{\dot{1}} \partial q^{\dot{1}}} \right)
\end{equation}
Arbitrariness in $q'^{\dot{1}} = q'^{\dot{1}} (q^{\dot{B}})$ allows to gauge away the function $N$. Gauge transformations are reduced then to the following formulas
\begin{equation}
q'^{\dot{2}} = q'^{\dot{2}} (q^{\dot{2}}) , \ \frac{d q'^{\dot{2}}}{d q^{\dot{2}}} =: \mu(q^{\dot{2}}) , \ 
q'^{\dot{1}} = \nu (q^{\dot{2}}) q^{\dot{1}} + g (q^{\dot{2}}) , \ \Delta = \mu \nu \ne 0
\end{equation}
where $g$, $\mu$ and $\nu$ are arbitrary gauge functions of one variable $q^{\dot{2}}$.

Eqs. (\ref{rownania_ostateczne_einsteinowskie}) with $M=N=0$ implies
\begin{equation}
P = f q^{\dot{1}} + \beta, \ S = f^{2} {q^{\dot{1}}}^{3} + \left( \frac{d f}{d q^{\dot{2}}} + 3f \beta \right) {q^{\dot{1}}}^{2} + \gamma q^{\dot{1}} + \delta
\end{equation}
where $f, \beta, \gamma$ and $\delta$ are arbitrary functions of $q^{\dot{2}}$. Transformation formulas for these functions read
\begin{eqnarray}
\label{wzory_transformacyjne_na_cztery}
\mu \nu f' &=& f
\\ \nonumber
\mu \beta' &=& \beta - \mu g f' + \frac{1}{3} \frac{d}{d q^{\dot{2}}} \ln (\mu \nu^{-2})
\\ \nonumber
\mu^{2} \nu^{3} \gamma' &=& \nu^{3} \gamma - 3 \nu g^{2} f^{2} - 2\nu^{3}\mu^{2}g \left( \frac{1}{\mu} \frac{d f'}{d q^{\dot{2}}} + 3 f'\beta' \right) -3 \nu^{2} \frac{d \nu}{d q^{\dot{2}}} \beta - 3 f \nu^{2} \frac{dg}{d q^{\dot{2}}}
\\ \nonumber
&& - \nu^{2} \frac{d^{2} \nu}{dq^{\dot{2}}dq^{\dot{2}}} + 2\nu \left( \frac{d \nu}{d q^{\dot{2}}} \right)^{2}
\\ \nonumber
\mu^{2}\nu^{2} \delta' &=& \nu^{3} \delta -3\nu^{2} \beta \frac{dg}{d q^{\dot{2}}} - \nu^{2} \frac{d^{2} g}{dq^{\dot{2}}dq^{\dot{2}}} + 2\nu \frac{dg}{d q^{\dot{2}}} \frac{d \nu}{d q^{\dot{2}}} - \mu^{2}\nu^{2} ( f'^{2}g^{3} + g \gamma')
\\ \nonumber
&& -\mu^{2}\nu^{2}g^{2} \left( \frac{1}{\mu} \frac{d f'}{d q^{\dot{2}}} + 3 f' \beta' \right)
\end{eqnarray}
From (\ref{wzory_transformacyjne_na_cztery}) it follows that one can always put $\beta = \gamma = \delta = 0$ without any loss of generality. Finally we obtain the solution
\begin{eqnarray}
\nonumber
&& A_{\dot{B}} = 0, \ B=B_{0} \ne 0, \ B_{\dot{A}\dot{C}} = 0, \ C_{\dot{A}}=0
\\ \nonumber
&& C_{\dot{1}\dot{1}\dot{1}}=0, \ C_{\dot{1}\dot{1}\dot{2}}=0, \ C_{\dot{1}\dot{2}\dot{2}}=f q^{\dot{1}}, \ C_{\dot{2}\dot{2}\dot{2}}=f^{2} {q^{\dot{1}}}^{3} + \frac{d f}{d q^{\dot{2}}} {q^{\dot{1}}}^{2}
\\ \nonumber
&& E_{\dot{1}\dot{1}}=0,\ E_{\dot{1}\dot{2}} = \frac{f}{B_{0}} , \ E_{\dot{2}\dot{2}} =  \frac{1}{B_{0}} \left( f^{2} {q^{\dot{1}}}^{2} + \frac{d f}{d q^{\dot{2}}} {q^{\dot{1}}} \right)
\\ \nonumber
\end{eqnarray}
Hence
\begin{eqnarray}
\label{Q_dla_drugiej_mozliwosci}
Q^{\dot{1}\dot{1}} &=& B_{0} p^{\dot{1}}p^{\dot{1}} - f q^{\dot{1}} p^{\dot{1}} - \left( f^{2} {q^{\dot{1}}}^{3} + \frac{d f}{d q^{\dot{2}}} {q^{\dot{1}}}^{2} \right)  p^{\dot{2}} + \frac{1}{B_{0}} \left( f^{2} {q^{\dot{1}}}^{2} + \frac{d f}{d q^{\dot{2}}} q^{\dot{1}} \right)
\\ \nonumber
Q^{\dot{1}\dot{2}} &=& B_{0} p^{\dot{1}}p^{\dot{2}}  + f q^{\dot{1}} p^{\dot{2}} - \frac{f}{B_{0}}
\\ \nonumber
Q^{\dot{2}\dot{2}} &=& B_{0} p^{\dot{2}}p^{\dot{2}}
\end{eqnarray}
where $f=f(q^{\dot{2}})$ is an arbitrary function. Using (\ref{oznaczenia_wspolrzednych}) and (\ref{wzor_na_B0}) one arrives at the metric
\begin{eqnarray}
\label{rozwiazanie_Dxnic_einsteinowskie}
\frac{1}{2} ds^{2} &=& dy dq - dx dp + \frac{\Lambda}{3} (y dq-x dp)^{2} - 2 f \left( qy- \frac{3}{\Lambda} \right) dp dq
\\ \nonumber
&& - \left( f qx + \left( qy- \frac{3}{\Lambda} \right) \left( f^{2} q^{2} + \frac{d f}{d p} q \right)   \right)  dp^{2}, \ \ f=f(p)
\end{eqnarray}
\textbf{Remark}: The problem of Einstein spaces of the type $[\textrm{D}]^{nn} \otimes [-]^{e}$ can be attacked more directly, by using the formalism of the hyperheavenly spaces instead of the weak hyperheavenly spaces. Frankly, we followed that way and we arrived exactly at the solution (\ref{rozwiazanie_Dxnic_einsteinowskie}). The \textsl{key function}, which generates the metric (\ref{rozwiazanie_Dxnic_einsteinowskie}) has the form
\begin{equation}
\label{funkcja_kluczowa_generujaca}
\Theta = \left( \frac{1}{2} f q y^{2} - \frac{3f}{\Lambda} y  \right) x + \frac{1}{6} \left( f^{2} q^{3} + \frac{df}{d p} q^{2} \right) y^{3}- \frac{3}{2 \Lambda} \left( f^{2} q^{2} + \frac{df}{d p} q \right) y^{2} 
\end{equation}

The metric (\ref{rozwiazanie_Dxnic_einsteinowskie}) contains one arbitrary function of one variable $f(p)$. This function does not enter into the curvature or higher curvature invariants what suggests that it is gauge-dependent. We were not able to find an explicit transformation of the variables which makes $f=0$. However, there are a little more subtle ways to prove that $f$ is not essential. We are indebted to the anonymous Referee for pointing out these ways for us.

First we consider the Killing vectors of the metric (\ref{rozwiazanie_Dxnic_einsteinowskie}). After some tedious calculations we arrive at the conclusion, that any Killing vector admitted by the metric (\ref{rozwiazanie_Dxnic_einsteinowskie}) can be brought to the form 
\begin{eqnarray}
K &=& (-faq^{3} + a_{p} q^{2} +c q +e) \, \frac{\partial}{\partial q} + (aq+b) \, \frac{\partial}{\partial p}
\\ \nonumber
&& + \bigg[ \left( - \partial_{p} (af) q^{3} + a_{pp} q^{2} + c_{p} q + e_{p} \right) y - (a_{p}q+b_{p}) x 
\\ \nonumber
&& \ \ \ - \frac{1}{\Lambda} \left( -3 \partial_{p} (af) q^{2} + 3 a_{pp} q +c_{p} + b_{pp} \right)   \bigg] \frac{\partial}{\partial x}
\\ \nonumber
&& + \left[  (3faq^{2} - 2 a_{p} q -c)y +ax + \frac{3}{\Lambda} (a_{p}-2faq   )  \right]  \frac{\partial}{\partial y}
\end{eqnarray}
where
\begin{eqnarray}
&& a (p) := eF F_{p} - \alpha_{0} F^{2} - F_{p} (m_{0} p^{2} + r_{0} p +s_{0}) + F (m_{0} p + r_{0} - n_{0}) + z_{0} p + j_{0}  \ \ \ \ \ \ \ \ 
\\ \nonumber
&& b (p) := -eF + m_{0} p^{2} + r_{0} p + s_{0}
\\ \nonumber
&& c (p) := \alpha_{0} F - 2eF_{p} + m_{0} p + n_{0}
\\ \nonumber
&& e (p) := \alpha_{0} p + \beta_{0}
\end{eqnarray}
$\alpha_{0}$, $\beta_{0}$, $m_{0}$, $n_{0}$, $r_{0}$, $s_{0}$, $z_{0}$ and $j_{0}$ are constants and $F(p)$ is a function such that $f=F_{pp}$. Hence, the metric (\ref{rozwiazanie_Dxnic_einsteinowskie}) admits 8 Killing vectors. There is only one SD Einstein metric with 8-dimensional symmetry algebra (see, e.g., \cite{Dunajski_Mettler}). It reads
\begin{equation}
\label{rozwiazanie_Dxnic_einsteinowskie_bez_funkcji_f}
\frac{1}{2} ds^{2} = dy dq - dx dp + \frac{\Lambda}{3} (y dq-x dp)^{2} 
\end{equation}
which is exactly the metric (\ref{rozwiazanie_Dxnic_einsteinowskie}) with $f=0$. It proves, that $f$ in (\ref{rozwiazanie_Dxnic_einsteinowskie}) can be put zero without any loss of generality. For completeness we list all Killing vectors admitted by the metric (\ref{rozwiazanie_Dxnic_einsteinowskie_bez_funkcji_f}). They read
\begin{eqnarray}
&& K_{1} = \frac{\partial}{\partial p}, \ K_{2} = \frac{\partial}{\partial q} , \ K_{3} = q \frac{\partial}{\partial q}  - y \frac{\partial}{\partial y}, \  K_{4} = p \frac{\partial}{\partial p} - x \frac{\partial}{\partial x}
\\ \nonumber
&& K_{5} = q \frac{\partial}{\partial p} + x \frac{\partial}{\partial y} , \ K_{6} = p \frac{\partial}{\partial q} + y \frac{\partial}{\partial x} 
\\ \nonumber
&& K_{7} =   q \left( q \frac{\partial}{\partial q}+p  \frac{\partial}{\partial p} - x \frac{\partial}{\partial x}  \right) + \left( \frac{3}{\Lambda} - 2qy +px \right) \frac{\partial}{\partial y} 
\\ \nonumber
&& K_{8} =  p \left( q \frac{\partial}{\partial q} + p \frac{\partial}{\partial p} - y \frac{\partial}{\partial y} \right) + \left( qy-2px - \frac{3}{\Lambda}   \right) \frac{\partial}{\partial x}
\end{eqnarray}
Note, that with $f=0$ the key function (\ref{funkcja_kluczowa_generujaca}) vanishes. It means that in the second Plebański formalism the general metric of the type $[\textrm{D}]^{nn} \otimes [-]^{e}$ is generated by the key function $\Theta=0$.

If coordinates $(q,p,x,y)$ are real and $f$ is a real smooth function then the metric (\ref{rozwiazanie_Dxnic_einsteinowskie}) becomes real metric with neutral signature and with nonzero curvature scalar. An interesting way of construction of such a metric has been presented in \cite{Dunajski_Mettler}. The authors of \cite{Dunajski_Mettler} constructed a neutral signature Einstein metric with $R \ne 0$ from a given projective structure on a 2-dimensional surface. Such a metric has a form
\begin{equation}
\label{metryka_Dunajski}
\frac{1}{2} ds^{2} = d \xi_{i} dx^{i} - \left( \Gamma^{k}_{\  ij} \xi_{k} - \frac{\Lambda}{3} \, \xi_{i} \xi_{j} -\frac{3}{\Lambda} P_{ij}  \right)   dx^{i} dx^{j}
\end{equation}
where $i,j,k=1,2$, $(\xi_{i}, x^{j})$ are local coordinates, $\Gamma^{k}_{\  ij}$ are the Christoffel symbols of a $[\nabla]$-representative connection and $P_{ij}$ is the Schouten tensor of $\Gamma^{k}_{\  ij}$. The metric (\ref{rozwiazanie_Dxnic_einsteinowskie}) fits in (\ref{metryka_Dunajski}) if we put
\begin{equation}
\nonumber
\xi_{1} = y, \ \xi_{2} = -x, \ x^{1} = q, \ x^{2} = p 
\end{equation}
Then
\begin{eqnarray}
&& \Gamma^{k}_{\  11} = \Gamma^{2}_{\  12}=0, \ \Gamma^{1}_{\  12} = f q, \ \Gamma^{2}_{\  22} = -f q, \ \Gamma^{1}_{\  22} =    f^{2} q^{3} + \frac{df}{d p} q^{2}  \ \ \ \ 
\\ \nonumber
&& P_{11} = 0, \ P_{12} = P_{21} = f, \ P_{22} = f^{2} q^{2} + \frac{df}{d p} q
\end{eqnarray}
The projective curvature of the projective structure in 2-dimensions is given by $\nabla_{[i}P_{j]k}$ which is zero for any $f$. Therefore the projective structure is projectively flat. Hence, the corresponding SD Einstein metric is isometric to (\ref{rozwiazanie_Dxnic_einsteinowskie_bez_funkcji_f}) what again prove, that $f$ in (\ref{rozwiazanie_Dxnic_einsteinowskie}) can be eliminated.

\subsection{Einstein spaces of the type $[\textrm{D}]^{nn} \otimes [-]^{e}$ in double null coordinates system}

\renewcommand{\arraystretch}{2.1}
\setlength\arraycolsep{2.5pt}

In this section we transform the solution (\ref{rozwiazanie_Dxnic_einsteinowskie_bez_funkcji_f}) into double null coordinates system. To do it we proceed, as follows.

The first congruence of SD null strings is given by the Pfaff system $e^3=0$, $e^1=0$ or, equivalently, $q^{\dot{A}} = \textrm{const}$. The second congruence is defined by the Pfaff system $e^4=0$, $e^2=0$ (because $\eth^{\dot{A}} Q_{\dot{A}\dot{B}}=0$). It means that there exist functions $\mathcal{F}_{A \dot{B}}$ and $q^{A}$ such that 
\begin{equation}
-dp^{\dot{A}} +Q^{\dot{A}\dot{B}} dq_{\dot{B}} = \mathcal{F}^{B \dot{A}} dq_{B}
\end{equation}
Hence, the second congruence of SD null strings is given by the equations $q^{A} = \textrm{const}$. Then
\begin{subequations}
\begin{eqnarray}
\label{metryka_standard_hiper}
\frac{1}{2} ds^{2} &=& -dp^{\dot{A}}  dq_{\dot{A}} + Q^{\dot{A}\dot{B}} \, dq_{\dot{A}}  dq_{\dot{B}} 
\\
\label{metryka_double_null}
&=&\mathcal{F}_{A \dot{B}} \, d q^{A} d q^{\dot{B}} 
\end{eqnarray}
\end{subequations}
To pass to the double null coordinates system the coordinates $q^{\dot{B}}$ should be kept unchanged and coordinates $p^{\dot{A}}$ should be treated as functions of $(q^{B}, q^{\dot{C}})$. It is equivalent to the set of equations 
\begin{equation}
\label{rownania_przejscia}
Q^{\dot{A}\dot{B}} = \frac{\partial p^{(\dot{A}}}{\partial q_{\dot{B})}}
\end{equation}
Then
\begin{equation}
\label{rownanie_na_potencjal}
\mathcal{F}_{A \dot{B}}  = \frac{\partial p_{\dot{B}}}{\partial q^{A}}
\end{equation}
with
\begin{eqnarray}
\nonumber
\det \mathcal{F}_{A\dot{B}} &=& 
   \left| \begin{array}{cc}
   \mathcal{F}_{1\dot{1}} & \mathcal{F}_{1\dot{2}}    \\  
   \mathcal{F}_{2\dot{1}} & \mathcal{F}_{2\dot{2}}   \\ 
    \end{array} \right|  =
\left| \begin{array}{cc}
   \dfrac{\partial y}{\partial r} & -\dfrac{\partial x}{\partial r}    \\  
   \dfrac{\partial y}{\partial s} & -\dfrac{\partial x}{\partial s}   \\ 
    \end{array} \right|   \ne 0
\end{eqnarray}
where we denoted
\begin{equation}
p^{\dot{1}} =:x, \ p^{\dot{2}} =: y, \ q^{\dot{1}} =:q, \ q^{\dot{2}} =: p, \ q^{1} =:r, \ q^{2} := s
\end{equation}
Hence, (\ref{metryka_double_null}) takes the form
 \begin{equation}
 \label{metryka_doubbllee_nulll_druga}
\frac{1}{2} ds^{2} =  \frac{\partial y}{\partial r} \, dr dq - \frac{\partial x}{\partial r} \, dr dp
+ \frac{\partial y}{\partial s} \, ds dq - \frac{\partial x}{\partial s} \, ds dp
\end{equation}
Eqs. (\ref{rownania_przejscia}) read
\begin{equation}
\label{rownanie_przejscia}
 \frac{\partial y}{\partial q} = -B_{0} y^{2} , \ 
 \frac{\partial y}{\partial p} - \frac{\partial x}{\partial q} = 2B_{0} xy , \ 
 \frac{\partial x}{\partial p} = B_{0} x^{2} ; \ \ 3 B_{0} = \Lambda
\end{equation}
with solution
\begin{equation}
x= - \frac{J}{B_{0} (B_{0}q+J p +G)} , \ y=  \frac{1}{ B_{0}q+J p +G}
\end{equation}
where $J$ and $G$ are arbitrary functions of $(r,s)$. Performing straightforward calculations one finds that the factors $G_{r}dr + G_{s} ds$ and $J_{r}dr + J_{s}ds$ appear in (\ref{metryka_doubbllee_nulll_druga}). These factors are obviously equal to $dG$ and $dJ$, respectively. It suggest that the functions $G$ and $J$ should be considered as the new coordinates. Changing the names of these coordinates, namely $G \rightarrow  s$, $J \rightarrow r$ we arrive at the metric 
\begin{equation}
\label{typ_DnnxDnn_einstein_IIformalizm}
\frac{1}{2}  ds^{2} =  \frac{3}{\Lambda(rp+s+q)^{2}} \big(  (q+s) \, drdp - p \, dr dq - r \, ds dp - ds dq  \big)
\end{equation}
\begin{comment}
Equivalently, the metric (\ref{typ_DnnxDnn_einstein_IIformalizm}) can be written in the form
\begin{equation}
\label{metryka_para_Kahler_nowa}
\frac{1}{2} ds^{2} = \mathcal{F}_{pr} \, dpdr + \mathcal{F}_{ps} \, dpds + \mathcal{F}_{qr} \, dqdr + \mathcal{F}_{qs} \, dqds
\end{equation}
with the potential $\mathcal{F}$ 
\begin{equation}
\label{Potencjal_F_3}
 \mathcal{F} = \frac{3}{\Lambda} \ln (rp+s+q) 
\end{equation}
\end{comment}
The metric (\ref{typ_DnnxDnn_einstein_IIformalizm}) can be brought to more plausible form by the transformation
\begin{equation}
r=\frac{1}{w}, \ p = \frac{1}{\widetilde{z}}, \ s = \frac{z}{w}, \ q = \frac{\widetilde{w}}{\widetilde{z}}
\end{equation}
Hence
\begin{equation}
\label{typ_DnnxDnn_einstein_IIformalizm_optymalna}
\frac{1}{2} ds^{2} =  \frac{3}{\Lambda(1+w \widetilde{w} + z \widetilde{z})^{2}}  \big( (1+ w \widetilde{w}) \, dz d \widetilde{z} + (1+  z \widetilde{z}) \, dw d \widetilde{w} - w \widetilde{z} \, dz d \widetilde{w} - z \widetilde{w} \, dw d \widetilde{z} \big)
\end{equation}
The metric (\ref{typ_DnnxDnn_einstein_IIformalizm_optymalna}) can be rewritten as
\begin{equation}
\label{typ_DnnxDnn_einstein_IIformalizm_optymalna_2}
\frac{1}{2} ds^{2} = \mathcal{F}_{w\widetilde{w}} \, dw d\widetilde{w} + \mathcal{F}_{w\widetilde{z}} \, dw d\widetilde{z} + \mathcal{F}_{z\widetilde{w}} \, dz d\widetilde{w} + \mathcal{F}_{z\widetilde{z}} \, dz d\widetilde{z}
\end{equation}
with the potential
\begin{equation}
\mathcal{F} = \frac{3}{\Lambda} \ln (1 + w \widetilde{w} + z \widetilde{z})
\end{equation}
Finally we arrive at the conclusion that the metric of the SD Einstein space of the type $[\textrm{D}]^{nn} \otimes [-]^{e}$ can be always brought to the form (\ref{typ_DnnxDnn_einstein_IIformalizm_optymalna}). In general this metric is complex and coordinates $(w, \widetilde{w}, z ,\widetilde{z})$ are also complex. From the complex metric (\ref{typ_DnnxDnn_einstein_IIformalizm_optymalna}) three different real slices can be obtained. Indeed:
\begin{enumerate}
\item For the neutral slice of the type $[\textrm{D}_{r}]^{nn} \otimes [-]^{e}$ coordinates $(w, \widetilde{w}, z ,\widetilde{z})$ are real. Such a space is equipped with two (real) nonexpanding congruences of SD null strings. 
\item For the neutral slice of the type $[\textrm{D}_{c}] \otimes [-]^{e}$ coordinates $(w, \widetilde{w}, z, \widetilde{z})$ are complex and such that $\widetilde{z} = \bar{w}$, $\widetilde{w} = \bar{z}$ where bar stands for the complex conjugation. In this case there are no real congruences of SD null strings.
\item For the Riemannian slice of the type $[\textrm{D}] \otimes [-]$ coordinates $(w, \widetilde{w}, z ,\widetilde{z})$ are complex and such that $\widetilde{w}=\bar{w}$ and $\widetilde{z} = \bar{z}$. In this case the metric (\ref{typ_DnnxDnn_einstein_IIformalizm_optymalna}) is Fubini-Study metric.
\end{enumerate}

The procedure presented in this section can be applied to the non-Einstein metrics (\ref{postac_dKS}) and (\ref{Non_Einstein_second_solution}), but the results will be presented elsewhere.

%#####################################################################################

\section{Concluding remarks.}

In this paper we have investigated SD spaces which are two-sided conformally recurrent. Although many papers have been devoted to the conformally recurrent spaces, there is only one paper which consider SD spaces with such a property \cite{Plebanski_Przanowski_rec}. In that paper Plebański and Przanowski were able to integrate type-[N] SD spaces. The question of the type-[D] SD spaces remained open. In the present paper we have filled this gap. 

We proved that two-sided conformally recurrent SD spaces of the type [D] belong to the special class $[\textrm{D}]^{nn} \otimes [-]^{e}$, i.e., they are equipped with two nonexpanding congruences of SD null strings. Conversely, any space of the type $[\textrm{D}]^{nn} \otimes [-]^{e}$ is a two-sided conformally recurrent space. This property is crucial and allows to integrate the equations. There are two classes of non-Einstein spaces (\ref{postac_dKS}) and (\ref{Non_Einstein_second_solution}) and one class of Einstein spaces (\ref{rozwiazanie_Dxnic_einsteinowskie_bez_funkcji_f}). The real neutral and Riemannian slices of the Einstein solution have been also found.

The further investigations should be focused on the two issues related to the non-Einstein solutions. The algebraic reason why the non-Einstein solutions split in two classes (\ref{postac_dKS}) and (\ref{Non_Einstein_second_solution}) is obvious. The geometrical reason is unknown. More detailed investigations of the properties of the traceless Ricci tensor could answer this question. We believe, that the authors of \cite{Plebanski_Przanowski_rec} faced the same problem. They arrived at the three classes of solutions of the type-[N] SD spaces. Also in this case the geometrical difference between these classes is unclear. The question how obtain real Riemannian slices of the non-Einstein complex metrics (\ref{postac_dKS}) and (\ref{Non_Einstein_second_solution}) is also opened. The work on it is underway.
\newline
\newline
\textbf{Acknowledgments.} We are indebted to the Referee of this paper for detailed review and for proving that the function $f$ in (\ref{rozwiazanie_Dxnic_einsteinowskie}) is gauge-dependent. We are also grateful to Maciej Przanowski for many discussions and his interest in our work.

%#####################################################################################

\end{document}